\documentclass{amsart}

\usepackage{amsmath}
\usepackage{amssymb}
\usepackage{amsthm}
\usepackage[mathscr]{eucal}
\usepackage{graphicx}
\usepackage{mathrsfs}
\usepackage{fullpage}
\usepackage{color}
\usepackage{fancyhdr}

\theoremstyle{plain}
\newtheorem{theorem}{Theorem}[section]
\newtheorem{definition}[theorem]{Definition}
\newtheorem{proposition}[theorem]{Proposition}
\newtheorem{lemma}[theorem]{Lemma}

\newtheorem{assumption}[theorem]{Assumption}
\newtheorem{example}[theorem]{Example}
\newtheorem{remark}[theorem]{Remark}

\numberwithin{equation}{section}
\numberwithin{theorem}{section}

\newcommand{\lb}{\left\{}
\newcommand{\rb}{\right\}}
\newcommand{\Def}{\overset{\text{def}}{=}}
\newcommand{\R}{\mathbb{R}}
\newcommand{\N}{\mathbb{N}}
\newcommand{\Err}{\mathcal{E}}

\newcommand{\vkap}{\varkappa}
\newcommand{\eps}{\varepsilon}
\newcommand{\Borel}{\mathscr{B}}
\newcommand{\Pspace}{\mathscr{P}}
\newcommand{\BP}{\mathbb{P}}
\newcommand{\BE}{\mathbb{E}}
\newcommand{\filt}{\mathscr{F}}
\newcommand{\dilt}{\mathscr{D}}
\newcommand{\Pint}{{\Pspace[0,1]}}

\newcommand{\OO}{\mathcal{O}}
\DeclareMathOperator{\dist}{dist}
\DeclareMathOperator{\supp}{supp}
\DeclareMathOperator{\dom}{\text{dom}}
\DeclareMathOperator{\ridom}{\text{ri dom}}
\newcommand{\XX}{\mathsf{X}}
\newcommand{\UU}{\mathsf{U}}
\newcommand{\KComp}{\mathcal{K}}
\newcommand{\ppp}{\mathsf{\gamma}}
\newcommand{\BB}{\mathcal{B}}
\newcommand{\SubProb}{\mathscr{M}_1(\XX)}
\newcommand{\pt}{\star}
\newcommand{\ZZ}{\mathcal{Z}}
\newcommand{\IRec}{\mathcal{R}}
\newcommand{\aff}{\text{aff}}
\newcommand{\qq}{\text{q}}
\newcommand{\wprho}{\varrho}

\newcommand{\esquare}{\begin{flushright}$\square$\end{flushright}}

\allowdisplaybreaks

\begin{document}
\title{Recovery Rates in investment-grade pools of credit assets:\\A large deviations analysis}
\author{Konstantinos Spiliopoulos}
\address{Division of Applied Mathematics\\
Brown University\\
Providence, RI 02912}
\email{kspiliop@dam.brown.edu}

\author{Richard B. Sowers}
\address{Department of Mathematics\\
    University of Illinois at Urbana--Champaign\\
    Urbana, IL 61801}
\email{r-sowers@illinois.edu}

\date{\today.}

\begin{abstract} We consider the effect of recovery rates on a pool of
credit assets. We allow the recovery rate to depend on the defaults in a general way.  Using the theory of large deviations, we study the structure of losses
in a pool consisting of a continuum of types. We derive the corresponding rate function and show that it has a natural interpretation as the favored way to rearrange recoveries and losses among the different types. Numerical examples are also provided.




\end{abstract}

\maketitle

\section{Introduction}

Understanding the behavior of \emph{large pools of credit assets} is currently a problem of central importance.  Banks often hold such large pools
and their risk-reward characteristics need to be carefully managed.  In many cases, the losses in the pool are (hopefully) \emph{rare} as a consequence of diversification.
In  \cite{SowersCDOI},
we have used the theory of large deviations to gain some insight
into several aspects of rare losses in pools of credit assets.  Our interest here is
the effect of \emph{recovery}.  While a creditor either defaults or
doesn't (a Bernoulli random variable), the amount recovered may in
fact take a continuum of values.  Although many models assume that
recovery rate is constant---i.e., a fixed deterministic percentage
of the par value, in reality the
statistics of the amount recovered should be a bit more complicated.
The statistics of the recovered amount should
depend on the number of defaults; a large number of defaults
corresponds to a bear market, in which case it is more difficult to
liquidate the assets of the creditors. Our goal is to understand how to include
this effect in the study of rare events in large pools.
We would like to look at these rare events via
some ideas from statistical mechanics, or more accurately the theory
of \emph{large deviations}.  Large deviations formalizes the idea
that nature prefers ``minimum energy'' configurations when rare
events occur.  We would like to see how these ideas can be used
in studying the interplay between default rate and recovery rate.

Our work is motivated by the general challenge of understanding the effects
of nonlinear interactions between various parts of complicated financial systems.
One of the strengths of the theory of large deviations is exactly that it allows one
to focus on propagation of rare events in networks.  Our interest here
is to see how this can be implemented in a model for recovery rates
which depend on the default rate.

This work is part of a growing body of literature which applies the theory of large deviations to problems of rare losses in credit assets and tries to approximate the tail distribution of total losses in a portfolio consisting of many positions.
For a survey on some existing large deviations methods applied to finance and credit risk see \cite{MR2384674}. In \cite{GKS07,GL05} rare event asymptotics for the loss distribution in the Gaussian copula model of portfolio credit risk and related  importance sampling questions are studied. In \cite{Gordy02} the author considers
saddle point approximations to the tails of the loss distribution.
Measures such as Value-at-Risk (VaR) and expected shortfall have been developed in order to characterize the risk of a portfolio as a whole, see for example \cite{DS03, Gordy03, MFE05,FM02,FM03}. Issues of recovery have also been considered in the works
\cite{ABRS05,Das-Hanouna} and in the references therein. Our work is most closely related to \cite{MR2022976, leijdekker-2009}, where the dynamics of a configuration
of defaults was studied.  In \cite{MR2022976} it is assumed that  a "macro-environmental" variable $Y$, to which all the obligors react, can be chosen so that conditional on $Y$, the recovery and default rates in a pool of finite number of "types" are independent. In  \cite{leijdekker-2009}  the recovery rate is assumed to be independent
of the defaults in a pool of one type with a brief treatment of a pool with finite types.  Our work here is explicitly interested in the \emph{dependence} of the recovery rate on the fraction of the number of defaults and the framework of our efforts is a continuum of
types.    The case of a continuum
of types requires slightly more topological sophistication. The dependence assumption and the continuum of types unavoidably complicate the proof since several auxiliary
technical results need to be proven.

The novelty of our result stems from the following two things.
The first novelty is that the distribution of the recoveries for the defaulted assets depends upon the number of defaults in a fairly general way.  This allows for consideration of the
case when recovery rates are affected by the number of defaults. Note however that
 the individual defaults are assumed independent. The second novelty is more
technical in nature. In particular, we prove the large deviations principle  for the joint family of random variables $(L_{N}, \nu_{N})$
where $L_N$ is the average loss and $\nu_N$ is the  empirical measure on type-space determined
by the names which default.  In addition to
the large deviations principle, we derive various equivalent  representations for the rate function which give some more insight to the favored way
to rearrange recoveries and losses among the different types and also ease its numerical computation.

The  paper is organized as follows. In Section \ref{S:Model} we introduce our model and establish some notation.
In Section \ref{S:TypicalEvents} we study the ``typical'' behavior of the loss of our pool; we need to understand this before we can identify
what behavior is ``atypical''. In Section \ref{S:Heuristics}, we present some formal sample calculations and  our main results, Theorems \ref{T:Main} and \ref{T:AlternativeExpression}. These calculations are indicative of the range
of possibilities and they illustrate the main results.  In Section \ref{S:ExamplesDiscussion} we present some numerical examples which
 illustrate some of the main aspects of our analysis and conclude with a discussion and remarks on future work.
 The proof of Theorem \ref{T:Main} is in Sections \ref{S:lsc}, \ref{S:Lower}
 and \ref{S:Upper}. Section \ref{S:AlternateExpression} contains an alternative expression of the rate function, which is a variational formula
which optimizes over all possible configuration of recoveries and defaults, and which leads to a Lagrange multiplier approach which can be numerically implemented.
Lastly, Section \ref{S:AuxiliaryLemmas} contains the proofs of several necessary technical results.

The model at the heart of our analysis is in fact very stylized.  Since our primary interest is the interaction between default rates and recovery rates, our model focuses on this effect, but simplifies a number of other effects.  In particular, we assume that the defaults themselves are independent.
Our work complements the existing literature  and hopefully contributes to our  understanding of the interplay and interaction of recovery and default rates. It seems plausible that more realistic models (e.g., which include
a systemic source of risk) can be analyzed by techniques which are extensions of those developed here.

\section{The Model}\label{S:Model}

In this section we introduce our model and explain our goals.
Let's start by considering a single bond (or ``name'').
For reference, let's assume that all bonds have par value of $\$1$.
If the bond defaults, the assets underlying the bond
are auctioned off and the bondholder recovers $r$ dollars, where $r\in [0,1]$.
We will record the default/survival coordinate as an element of $\{0,1\}$, where $1$ corresponds
to a default and $0$ to survival.

  The pool suffers a loss
when a bond defaults, and the amount of the loss is $\$1-r$, where $r\in[0,1]$ is the recovery amount (in dollars). Define
\begin{equation*}
 \Delta_{n} \Def \begin{cases} 1 &\text{if the $n^{\text{th}}$ name has defaulted}\\
0 &\text{otherwise} \end{cases}
\end{equation*}
and let $\ell_{n}$ be the non-zero loss that occurs in the event of default of the $n^{\text{th}}$ name, i.e., when $\Delta_{n}=1$.
Equivalently, $\ell_{n}$ is the non-zero exposure of the $n^{\text{th}}$ name that occurs when $\Delta_{n}=1$. Clearly if $r_{n}$ is the recovery of the $n^{\text{th}}$ name
then $\ell_{n}=1-r_{n}$.

The default and loss rates in the pool are then
\begin{equation}\label{Eq:DefaultRecoveryRates}
 D_N\Def \frac{1}{N}\sum_{n=1}^N\Delta_n \qquad \text{and}\qquad L_N \Def  \frac{1}{N}\sum_{n=1}^N\ell_n \Delta_{n}.
\end{equation}

We denote by $p^{N,n}$ the risk-neutral probability
of default of the $n^{\text{th}}$ name in a pool of $N$ names. Moreover, let $\wprho^{N,n}(D_{N},dr)$ be the distribution of the recovery of the $n^{\text{th}}$ name in a pool of $N$ names. Notice that we allow $\wprho^{N,n}(D_{N},dr)$ to depend on the default rate $D_{N}$.

Let us make now the aforementioned discussion rigorous by introducing the probability space that we consider.
The minimal state space for a single bond is the set $E_\circ \Def \{0\}\cup \left(\{1\}\times [0,1]\right)$.
Since we want to
consider a pool of bonds, the state space in our
model will be $E \Def E_\circ^\N$, where $\N \Def \{1,2\dots \}$.
$E_\circ$, as a closed subset of $\{0,1\}\times [0,1]$
is also Polish, and thus $E$ is also Polish.  We endow $E$ with the natural $\sigma$-algebra $\filt \Def \Borel(E)$.


The only remaining thing to specify is a probability measure on $(E,\filt)$.  For each $N\in \N$, fix $\{p^{N,n}: n\in \{1,2\dots N\}\}\subset [0,1]$.  These are the risk-neutral
default probabilities of the names
when the pool has $N$ names.  We next fix $\{\wprho^{N,n}: n\in \{1,2\dots N\}\}\subset C([0,1];\Pint)$\footnote{We shall write $\wprho\in C([0,1];\Pint)$ as a map from $[0,1]\times \Borel[0,1]$ into $[0,1]$
such that for each $D\in [0,1]$, $A\mapsto \wprho(D,A)$ is an element of $\Pint$ and such that $D\mapsto \wprho(D,\cdot)$ is weakly continuous.}; i.e., a collection of probability
measures on $[0,1]$, the range of the recovery,  indexed by the default rate $D_{N}$.  Namely, for each $N\in\N$ and $n\in\{1,2,\cdots,N\}$, $\wprho^{N,n}$ is a probability measure that depends on the default rate $D_{N}$. Some concrete examples are in Subsection  \ref{SS:NumericalExamples}.

For each $n\in \N$ and $\omega=(\omega_1,\omega_2\dots )$, define the coordinate random variable
$X_n(\omega)=\omega_n$.  For each $N\in \N$,
we then fix our risk-neutral probability measure $\BP_N\in \Pspace(E)$ (with associated expectation
operator $\BE_N$) by requiring that
\begin{equation*} \BE_N\left[\prod_{n=1}^Nf_n(X_n)\right]=\BE_N\left[\prod_{n=1}^N\lb (1-p^{N,n})f_n(0)+ p^{N,n}\int_{r\in [0,1]}f_n(1,r)\wprho^{N,n}(D_N,dr)\rb\right]\end{equation*}
for all $\{f_n\}_{n=1}^N\subset B(E_\circ)$.  In particular, the aforementioned construction implies that for each $n\in \N$,
$\{\Delta_n\}_{n=1}^N$ is an independent collection of random variables with $\BP_N\{\Delta_N=1\}=p^{N,n}$
under $\BP_N$.


With this probabilistic structure in place, we will clearly want to be able to condition on the default rate so that we can then focus on the recovery rates.
To this end we define the $\sigma$-algebra
$$\dilt \Def \sigma\{\Delta_n: n\in \N\}.$$

We are interested in the \emph{typical behavior} and \emph{rare events} in this system.  We seek to understand the structure of these rare events in our model.
We are interested in the following questions:
\begin{itemize}
\item What is the typical behavior of the system? In other words we want to characterize $\bar L\Def \lim_{N\to \infty}L_N$.
\item What can we say about rare events in this general setting? In particular we  compute the asymptotics of $\BP_N\{L_N\ge l\}$ as $N\to \infty$, particularly for $l>\bar L$. Then $\{L_n\ge l\}$ is an ``atypical'' event.
\end{itemize}
Understanding the structure of rare events may give guidance in how to optimally control how rare events propagate in large financial networks.

\begin{remark}
Our formulation  allows for a fairly general dependence of the recovery distribution on the default rate. This is also partially exploited in Subsection
 \ref{SS:NumericalExamples} where we study several concrete examples. On the other hand, the general treatment of this paper unavoidably complicates the proof of the
 law of large numbers and especially of the large deviations principle.
\end{remark}

\section{Typical Events: A Law of Large Numbers}\label{S:TypicalEvents}

Let's start our analysis by identifying the ``typical'' behavior of $L_N$ as $N\to \infty$ defined by equation (\ref{Eq:DefaultRecoveryRates}). In order to motivate the calculations, we first
investigate how  $D_N$ and $L_N$ behave as $N\to \infty$ in a homogeneous pool, Example \ref{Ex:onetype}, and in a heterogeneous pool of two types, Example \ref{Ex:twotypes}.

\begin{example}[Homogeneous Example]\label{Ex:onetype}
Fix $p\in [0,1]$ and $\wprho\in C([0,1];\Pint)$ and let $p^{N,n}=p$
and $\wprho^{N,n}=\wprho$ for all $N\in \N$ and $n\in \{1,2\dots N\}$.

Let's see what $D_{N}$ and $L_N$ look like in this setting.  By the law of large numbers $\lim_{N\to \infty}D_N=p$.  Thus in distribution
\begin{equation*}L_N \approx \frac{1}{N}\sum_{1\le n\le pN} (1-\xi_n) \end{equation*}
where the $\xi_n$'s are i.i.d. with distribution $\wprho(p,\cdot)$.
We should consequently have that
\begin{equation*} \lim_{N\to \infty}L_N = p\int_{r\in [0,1]}(1-r) \wprho(p,dr). \end{equation*}
\esquare
\end{example}

\begin{example}[Heterogeneous Example]\label{Ex:twotypes}  Fix $p_A$ and $p_B$ in $[0,1]$ and fix
$\wprho_A$ and $\wprho_B$ in $C([0,1];\Pint)$.
Every third bond will be of type $A$ and have default probability $p_A$
and recovery distribution
governed by $\wprho_A$, and the remaining bonds will have
default probability $p_B$ and recovery distribution governed by $\wprho_B$. In other words,   the rate of default among the $A$ bonds is $p_A$
and the rate of default among the $B$ bonds is $p_B$.

For future reference, let's separate the defaults and into the
 two types.  If we define
\begin{equation*} D^A_N \Def \frac{1}{\lfloor N/3\rfloor}\sum_{\substack{1\le n\le N \\ n\in 3\N}}\Delta_n \qquad \text{and}\qquad D^B_N \Def \frac{1}{N-\lfloor N/3\rfloor}\sum_{\substack{1\le n\le N \\ n\not \in 3\N}}\Delta_n.\end{equation*}
we have
\begin{equation*} D_N = \frac{\lfloor N/3\rfloor}{N}D^A_N+\frac{N-\lfloor N/3\rfloor}{N}D^B_N \approx \frac13 D^A_N + \frac23 D^B_N. \end{equation*}
Thus $\lim_{N\to \infty}D^A_N=p_A$ and $\lim_{N\to \infty}D^B_n=p_B$, so
\begin{equation*} \lim_{N\to \infty} D_N = \bar D \Def \frac{p_A}{3} + \frac{2p_B}{3}. \end{equation*}
Thus we should roughly have
\begin{equation*}L_N \approx \frac{1}{N}\lb \sum_{1\le n\le p_AN/3}(1-\xi^A_n)
+\sum_{1\le n\le 2p_BN/3}(1-\xi^B_n)\rb \end{equation*}
where the
$\xi^A_n$'s have law $\wprho_A(\bar D,\cdot)$, the $\xi^B_n$'s have
distribution $\wprho_B(\bar D,\cdot)$ and the $\xi^A_n$'s and $\xi^B_n$'s
are all independent.  Consequently, it seems natural that
\begin{equation*} \lim_{N\to \infty}L_N = \frac{p_A}{3}\int_{r\in [0,1]}(1-r) \wprho_A(\bar D,dr)+\frac{2p_B}{3}\int_{r\in [0,1]}(1-r)\wprho_B(\bar D,dr), \end{equation*}
the first term being the limit of the losses from the type A names,
and the second term being the limit of the losses from the type B names.
\esquare
\end{example}

In view of our examples, it seems reasonable that we should be able
to describe the average loss in the pool in terms of a frequency count
of
\begin{equation*}
\ppp^{N,n}\Def (p^{N,n},\wprho^{N,n}).
\end{equation*}
 We note that $\ppp^{N,n}$ takes
values in the set $\XX \Def [0,1]\times C([0,1];\Pint)$.  Since $\Pint$ is Polish,
so is $C([0,1];\Pint)$, and thus $\XX$ is Polish.
We will henceforth refer to elements of $\XX$ as \emph{types}.
For each $N\in \N$, we now define $\UU_N\in \Pspace(\XX)$ as
\begin{equation*} \UU_N \Def \frac{1}{N}\sum_{n=1}^N\delta_{\ppp^{N,n}}. \end{equation*}

Assumption \ref{A:LimitExists} below is our main standing assumption. It is a natural assumption since it makes the family
$\{\UU_{N}\}_{N\in\mathbb{N}}$ relatively compact. It will be assumed throughout the paper, even though this may not be always stated explicitly.

\begin{assumption}\label{A:LimitExists} We assume that $\UU\Def \lim_{N\to \infty}\UU_N$ exists \textup{(}in $\Pspace(\XX)$\textup{)}.
\esquare
\end{assumption}
\begin{remark} In the case of Example \ref{Ex:onetype}, we have that $\UU = \delta_{(p,\wprho)}$, while in the case of Example \ref{Ex:twotypes}, we have that
$\UU = \frac13 \delta_{(p_A,\wprho_A)} + \frac23 \delta_{(p_B,\wprho_B)}$.
\end{remark}

We can now identify the limiting behavior of $L_N$.  Define
\begin{equation}\label{E:typdefrate} \bar D \Def \int_{\ppp=(p,\wprho)\in \XX} p\UU(d\ppp) \qquad \text{and}\qquad
\bar L \Def \int_{\ppp=(p,\wprho)\in \XX} \lb p\int_{r\in [0,1]}(1-r)\wprho(\bar D,dr)\rb \UU(d\ppp). \end{equation}
To see that these quantities are well-defined, note that
\begin{equation}\label{E:contmaps} (p,\wprho)\mapsto p \qquad \text{and}\qquad \nu\mapsto \int_{r\in [0,1]}(1-r)\nu(dr) \end{equation}
are continuous mappings from, respectively, $\XX$ and $\Pint$, to $[0,1] \subset \R$.  The continuity of the first map of \eqref{E:contmaps} implies that $\bar D$
is well-defined.  Combining the continuity of both maps of \eqref{E:contmaps}, we get that the map
\begin{equation*} (p,\wprho) \mapsto p\int_{r\in [0,1]}(1-r)\wprho(\bar D,dr) \end{equation*}
is a continuous map from $\XX$ to $[0,1]\subset \R$; thus $\bar L$ is also well-defined.

\begin{lemma}\label{L:LLN}Let Assumption \ref{A:LimitExists} hold and consider $\bar{L}$ given by (\ref{E:typdefrate}). Then, for each $\eps>0$, we have that
\begin{equation*} \lim_{N\to \infty}\BP_N\lb \left|L_N-\bar L\right|\ge \eps\rb = 0. \end{equation*}
\esquare
\end{lemma}
\begin{proof}  For $\wprho\in C([0,1];\Pint)$ and $D\in [0,1]$, let's first define
\begin{equation*} \Gamma(D,\wprho) \Def \int_{r\in [0,1]}(1-r)\wprho(D,dr). \end{equation*}
Again we can use \eqref{E:contmaps} and show, by the same arguments used to show that $\bar D$ and $\bar L$ of \eqref{E:typdefrate}
are well-defined, that $\Gamma$ is well-defined, and furthermore that it is continuous on $[0,1]\times C([0,1];\Pint)$.
For each $N\in \N$, define
\begin{gather*} \hat{L}_N \Def \frac{1}{N}\sum_{n=1}^N \Delta_n\Gamma(D_N,\wprho^{N,n}), \qquad
\bar D_N \Def \frac{1}{N}\sum_{n=1}^N p^{N,n} = \int_{\ppp=(p,\wprho)\in \XX}p\UU_N(d\ppp)\\
\bar L_N \Def \frac{1}{N}\sum_{n=1}^N p^{N,n}\Gamma(\bar D,\wprho^{N,n}) = \int_{\ppp=(p,\wprho)\in \XX}p\Gamma(\bar D,\wprho)\UU_N(d\ppp); \end{gather*}
note that $\hat{L}_N$ is a random variable but $\bar D_N$ and $\bar L_N$
are deterministic.  Note also
that by weak convergence,
\begin{equation*}
\lim_{N\to \infty}\bar L_N = \bar L \text{ and } \lim_{N\to \infty}\bar D_N=\bar D.
\end{equation*}
The claim will follow if we can prove that
\begin{equation}\label{E:goal} \lim_{N\to \infty}\BE_N\left[|L_N-\hat{L}_N|^2\right]=0 \qquad \text{and}\qquad \lim_{N\to \infty}\BE_N\left[|\hat{L}_N-\bar L_N|^2\right]=0. \end{equation}

To see the first part of \eqref{E:goal}, we calculate that
\begin{equation*} L_N-\hat{L}_N = \frac{1}{N}\sum_{n=1}^N \left(\ell_n-\Gamma(D_N,\wprho^{N,n})\right)\Delta_n. \end{equation*}
Conditioning on $\dilt=\sigma\{\Delta_n: n\in \N\}$, we have that
\begin{equation*} \BE_N\left[\left(L_N-\hat L_N\right)^2\Big|\dilt\right]
= \frac{1}{N^2}\sum_{n=1}^N \lb \int_{r\in [0,1]}\left((1-r)-\Gamma(D_N,\wprho^{N,n})\right)^2\wprho^{N,n}(D_N,dr)\rb \Delta_n^2
\le \frac{1}{N}. \end{equation*}
This implies the first part of \eqref{E:goal}.

To see the second part of \eqref{E:goal}, we write that $\hat{L}_N - \bar L_N = \Err^1_N + \Err^2_N$ where
\begin{align*} \Err^1_N &= \frac{1}{N}\sum_{n=1}^N \left[ \Gamma(D_N,\wprho^{N,n})-\Gamma(\bar D_N,\wprho^{N,n})\right] \Delta_n \\
\Err^2_N &= \frac{1}{N}\sum_{n=1}^N \Gamma(\bar D_N,\wprho^{N,n})\left[ \Delta_N-p^{N,n}\right] \end{align*}
We first calculate that
\begin{equation*}\BE_N\left[\left|\Err^2_N\right|^2\right] = \frac{1}{N^2}\sum_{n=1}^N \left(\Gamma(\bar D_N,\wprho^{N,n})\right)^2p^{N,n}\left(1-p^{N,n}\right) \le \frac{1}{4N}. \end{equation*}
Thus $\lim_{N\to \infty}\BE_N\left[|\Err^2_N|^2\right]=0$.  Similarly,
\begin{equation*}\BE_N\left[\left|D_N-\bar D_N\right|^2\right] = \frac{1}{N^2}\sum_{n=1}^N p^{N,n}\left(1-p^{N,n}\right) \le \frac{1}{4N}. \end{equation*}
Therefore
\begin{equation*}\lim_{N\to\infty}\BE_N\left[\left|D_N-\bar D\right|^2\right]=0. \end{equation*}
To bound $\Err^1_N$, fix $\eta>0$. Due to Assumption \ref{A:LimitExists},   Prohorov's theorem implies that
$\{\UU_N;\, N\in \N\}$ is tight, so there is a $K_\eta\subset \subset C([0,1];\Pint)$ such that
\begin{equation*} \sup_{n\in \N} \UU_N([0,1]\times K_\eta^c)<\eta. \end{equation*}
Defining
\begin{equation*}  \omega_\eta(\delta) \Def \sup_{\substack{\wprho\in K_\eta \\D_1, D_2\in [0,1]\\ |D_1-D_2|<\delta}}\left|\Gamma(D_1,\wprho)-\Gamma(D_2,\wprho)\right| \end{equation*}
for all $\delta>0$, compactness of $K_\eta$ and $[0,1]$ and continuity of $\Gamma$ imply that $\lim_{\delta \searrow 0} \omega_\eta(\delta)=0$.  Thus
\begin{multline*} |\Err^1_N| \le \BE_N\left[\frac{1}{N}\sum_{n=1}^N\left|\Gamma(D_N,\wprho^{N,n})-\Gamma(\bar D_N,\wprho^{N,n})\right|\right]\\
= \BE_N\left[\int_{\ppp=(p,\wprho)\in \XX}\left|\Gamma(D_N,\wprho)-\Gamma(\bar D_N,\wprho)\right|\UU_N(d\ppp)\right] \\
\le 2\eta + 2 \omega_\eta(\delta) + \tilde \BP_N\lb |D_N-\bar D|\ge \delta\rb. \end{multline*}
Take $N\nearrow \infty$, then let $\delta \searrow 0$.  Finally let $\eta\searrow 0$.  We conclude that indeed $\lim_{N\to \infty}\BE_N[|\Err_N^1|]=0$.
Hence, the second part of (\ref{E:goal}) also holds which concludes the proof of the lemma.\end{proof}

\section{Problem Formulation and Main Results}\label{S:Heuristics}

Let's now set up our framework for considering atypical behavior of the $L_N$'s; i.e., large deviations. We motivate the main result, Theorem \ref{T:Main} through a formal discussion
related to the setting of Examples \ref{Ex:onetype} and \ref{Ex:twotypes}. The proof of Theorem \ref{T:Main} is in
 Sections \ref{S:lsc}, \ref{S:Lower}  and \ref{S:Upper}. In Section \ref{S:lsc} we prove compactness of the level sets of the rate function. In Sections \ref{S:Lower} and \ref{S:Upper} we prove
the large deviations lower and upper limit respectively. An equivalent insightful representation of the rate function is given in Theorem \ref{T:AlternativeExpression} and its proof is
in Section \ref{S:AlternateExpression}. This  representation   is a variational formula
which optimizes over all possible configuration of recoveries and defaults, and which leads to a Lagrange multiplier approach which can be numerically implemented.

For the convenience of the reader we here recall the concept of the large deviations principle and the associated rate function.
\begin{definition}\label{Def:LDP}
If $X$ is a Polish space and $\BP$ is a probability measure on $(X,\Borel(X))$, we say that a collection $(\xi_n)_{n\in \N}$
of $X$-valued random variables has a large deviations principle with rate function $I:X\to [0,\infty]$ if
\begin{enumerate}
\item For each $s\ge 0$, the set $\Phi(s) \Def \lb x\in X: I(x)\le s\rb$
is a compact subset of $X$.
\item For every open $G\subset X$,
\begin{equation*}
\varliminf_{n\nearrow\infty}\frac{1}{n}\ln \BP\lb \xi_n\in G\rb \geq -\inf_{x\in G} I(x)
\end{equation*}
\item For every closed $F\subset X$,
\begin{equation*}
\varlimsup_{n\nearrow\infty}\frac{1}{n}\ln \BP\lb \xi_n\in F\rb\leq -\inf_{x\in F} I(x).
\end{equation*}
\end{enumerate}
\end{definition}

One thing which is clear from Example \ref{Ex:twotypes} is that we need to keep track of the type associated with each default (but not the types associated
with names which do not default).  To organize this, let $\SubProb$ be the collection of measures $\nu$ on $(\XX,\Borel(\XX))$ such that $\nu(\XX)\le 1$ (i.e.,
the collection of subprobability measures).

As it is discussed in Section \ref{S:lsc}, $\SubProb$ is a Polish space.  For each $N\in \N$, define a random element $\nu_N$ of $\SubProb$
as
\begin{equation}
 \nu_N \Def \frac{1}{N}\sum_{n=1}^N \Delta_n\delta_{\ppp^{N,n}}\label{Eq:EmpiricalMeasure}
 \end{equation}
Thus, $\nu_N $ is the  empirical measure on type-space determined by the names which default.

Next, define a sequence $(Z_N)_{N\in \N}$ of $[0,1]\times \SubProb$-valued random variables as
\begin{equation*}  Z_N\Def (L_N,\nu_N) \qquad N\in \N. \end{equation*}

Since $[0,1]$ and $\SubProb$ are both Polish, $[0,1]\times \SubProb$ is also Polish.  We seek a large deviations principle for the $Z_N$'s.  Since projection
maps are continuous, the contraction principle will then imply a large deviations principle for the $L_N$'s.  Note furthermore that
\begin{equation*} D_N =\nu_N(\XX); \end{equation*}
the map $\nu\mapsto \nu(\XX)$ is continuous in the topology on
$\SubProb$, so the recovery statistics depend continuously on
$\nu_N$.

Before proceeding with the analysis, we need to define some more quantities.
 \begin{definition}\label{E:LFT}  For $p\in [0,1]$, define
\begin{equation*}  \begin{aligned} \hbar_p(x) &\Def \begin{cases} x\ln \frac{x}{p} + (1-x)\ln \frac{1-x}{1-p} &\text{for $x$ and $p$ in $(0,1)$} \\
\ln \frac{1}{p} &\text{for $x=1$, $p\in (0,1]$} \\
\ln \frac{1}{1-p} &\text{for $x=0$, $p\in [0,1)$} \\
\infty &\text{else.}\end{cases} \end{aligned}\end{equation*}
For $\nu\in \SubProb$, define
 \begin{equation*}H(\nu) \Def \begin{cases} \int_{\ppp=(p,\wprho) \in \XX}\hbar_p\left(\frac{d\nu}{d\UU}(\ppp)\right)\UU(d\ppp) &\text{if $\nu \ll \UU$} \\
 \infty. &\text{else}\end{cases}\end{equation*}
\esquare
\end{definition}

To understand the behavior of $L_N$, we need to construct its moment generating function. For this purpose, we have the following definition.
 \begin{definition}\label{Def:MomentGeneratingFunction} For $\nu\in \SubProb$, define
 \begin{align*} \Lambda_\nu(\theta) &\Def \int_{\ppp=(p,\wprho)\in \XX} \lb \ln \int_{r\in [0,1]}e^{\theta(1-r)}\wprho(\nu(\XX),dr)\rb \nu(d\ppp) \qquad \theta\in \R \\
 \Lambda^*_\nu(\ell) &\Def \sup_{\theta\in \R}\lb \theta \ell - \Lambda_\nu(\theta)\rb \qquad \ell\in [0,1]\end{align*}
$\Lambda^*_\nu(\ell)$ is the Legendre-Fenchel transform of $\Lambda_\nu(\theta)$.
 \esquare
 \end{definition}

Let's now see if we can formally identify a large deviations principle for the $Z_N$'s in the setting of Examples \ref{Ex:onetype} and \ref{Ex:twotypes}. In both examples we want to find a map
$I:[0,1]\times \SubProb\to [0,\infty]$ such that, at least formally,
\begin{equation*} \BP_N\lb Z_N\in dz^*\rb \asymp \exp\left[-N I(z^*)\right] \qquad N\to \infty \end{equation*}
for each fixed $z^*= (\ell^*,\nu^*)\in [0,1]\times \SubProb$. Assume that $\nu^* \ll \UU$. This is done without loss of generality, since, as we shall also see in the sequel,
the rate function $I(z^*)=\infty$ if $\nu^*\not \ll \UU$.

Conditioning on the event that $\nu_N\approx \nu^*$, we have
\begin{equation}\label{E:combine} \BP_N\lb Z_N\in dz^*\rb = \BP\lb L_N\in d\ell^*\big| \nu_N\approx \nu^*\rb \BP_N\lb \nu_N\in d\nu^*\rb. \end{equation}
We want to derive asymptotic expressions for the first and second term in the right hand side of \eqref{E:combine}.
\begin{example}\label{Ex:onetypeLDP}[Homogeneous Example \ref{Ex:onetype} continued].
First, we understand the second term in \eqref{E:combine}, i.e.,  $\BP_N\lb \nu_N\in d\nu^*\rb$.
Recall that
\begin{equation*} \nu_N(\XX) = \frac{1}{N}\sum_{\substack{1\le n\le N }}\Delta_n. \end{equation*}
This is now essentially the focus of Sanov's theorem--i.i.d. Bernoulli random variables. Namely, recalling the definition of $\hbar_p(x)$ and $H(\nu)$ by Definition \ref{E:LFT} we have
 \begin{equation} \label{E:optentHom} \BP_N\{\nu_N\in d\nu^*\} \asymp \exp\left[-N\hbar_{p}\left(\nu^*(\XX)\right)\right]=\exp\left[-N H\left(\nu^*(\XX)\right)\right]\end{equation}

Second, we  understand the first term in (\ref{E:combine}), i.e. $\BP\lb L_N\in d\ell^*\big| \nu_N\approx \nu^*\rb$.
 As in  Example \ref{Ex:twotypes}, we should have
 in law
 \begin{equation}\label{E:lossapproxHom} L_N \approx \frac{1}{N}\lb \sum_{1\le n\le N\nu^*(\XX)}(1-\xi_n)  \rb \end{equation}
 where the $\xi_n$'s are i.i.d. with common law $\wprho(\nu(\XX),\cdot)$. In the homogeneous pool of this example,
the log moment generating function defined in Definition \ref{Def:MomentGeneratingFunction} takes the form
\begin{equation*} \Lambda_{\nu{^*}}(\theta) = \nu^*(\XX)\ln \int_{r\in[0,1]} e^{\theta(1-r)}\wprho(\nu(\XX),dr)  \end{equation*}
 for all $\theta\in \R$. The logarithmic moment generating function of $L_N$ of \eqref{E:lossapproxHom} is approximately $ N\Lambda_{\nu^*}(\theta/N)$.
 Thus we should have that
 \begin{equation*} \BP_N\lb L_N\in d\ell^*\big| \nu_N\approx \nu^*\rb
 \asymp \exp\left[-N \Lambda^*_{\nu^{*}}(\ell)\right] \end{equation*}
 We should then get the large deviations principle for $Z_N$ by combining this, \eqref{E:optentHom}, and \eqref{E:combine}, i.e.,
\begin{equation*}\ \BP_N\lb Z_N\in dz^*\rb \asymp \exp\left[-N [\Lambda^*_{\nu^{*}}(\ell)+H(\nu^{*})]\right]. \end{equation*}

\esquare
\end{example}

\begin{example}\label{Ex:twotypesLDP}[Heterogeneous Example \ref{Ex:twotypes} continued].
We proceed similarly to Example \ref{Ex:onetypeLDP}. First, we understand the second term in \eqref{E:combine}, i.e., $\BP_N\lb \nu_N\in d\nu^*\rb$. Observe, that we explicitly have
\begin{equation*} \frac{d\nu^*}{d\UU}(\ppp_A) = \frac{\nu^*\{\ppp_A\}}{1/3} \qquad \text{and}\qquad \frac{d\nu^*}{d\UU}(\ppp_B) = \frac{\nu^*\{\ppp_B\}}{2/3}. \end{equation*}
Similarly
\begin{align*} \frac{d\nu_N}{d\UU}(\ppp_A) &= \frac{\nu_N\{\ppp_A\}}{\UU\{\ppp_A\}} = \frac{1}{N/3}\sum_{\substack{1\le n\le N \\ n\in 3\N}}\Delta_n \\
 \frac{d\nu_N}{d\UU}(\ppp_B) &= \frac{\nu_N\{\ppp_B\}}{\UU\{\ppp_B\}} = \frac{1}{2N/3}\sum_{\substack{1\le n\le N \\ n\not \in 3\N}}\Delta_n. \end{align*}
 Thus
 \begin{multline*} \BP_N\{\nu_N\in d\nu^*\} = \BP_N\lb \nu_N\{\ppp_A\} \approx \nu^*\{\ppp_A\},\,  \nu_N\{\ppp_B\} \approx \nu^*\{\ppp_B\}\rb\\
 = \BP_N\lb \frac{d\nu_N}{d\UU}(\ppp_A) \approx \frac{d\nu^*}{d\UU}(\ppp_A),\, \frac{d\nu_N}{d\UU}(\ppp_B) \approx \frac{d\nu^*}{d\UU}(\ppp_B)\rb. \end{multline*}

Again, this is  essentially the focus of Sanov's theorem--i.i.d. Bernoulli random variables. Namely, recalling the definition of $\hbar_p(x)$ and $H(\nu)$ by Definition \ref{E:LFT} we have
 \begin{equation} \label{E:optent}\begin{aligned} \BP_N\{\nu_N\in d\nu^*\} &\asymp \exp\left[-\frac{N}{3}\hbar_{p_A}\left(\frac{d\nu^*}{d\UU}(\ppp_A)\right)-\frac{2N}{3}\hbar_{p_B}\left(\frac{d\nu^*}{d\UU}(\ppp_B)\right)\right]\\
 &= \exp\left[-N\int_{\ppp=(p,\wprho)\in \XX}\hbar_p\left(\frac{d\nu^*}{d\UU}(\ppp)\right)\UU(d\ppp)\right]=\exp\left[-N H(\nu^{*})\right]. \end{aligned}\end{equation}

Second, we  understand the first term in (\ref{E:combine}), i.e. $\BP\lb L_N\in d\ell^*\big| \nu_N\approx \nu^*\rb$.
 If $\nu_N\approx \nu^*$, then there are about $N\nu^*\{\ppp_A\}$ defaults of type A, and $N\nu^*\{\ppp_B\}$ defaults of type B.  Thus, as in  Example \ref{Ex:twotypes}, we should have
 in law
 \begin{equation}\label{E:lossapprox} L_N \approx \frac{1}{N}\lb \sum_{1\le n\le N\nu^*\{\ppp_A\}}(1-\xi^A_n) +\sum_{1\le n\le N\nu^*\{\ppp_B\}}(1-\xi^B_n) \rb \end{equation}
 where the $\xi^A_n$'s are i.i.d. with common law $\wprho_A(\nu(\XX),\cdot)$ and the $\xi^B_n$'s are i.i.d. with common law $\wprho_B(\nu(\XX),\cdot)$. In the heterogeneous pool of this example,
the log moment generating function defined in Definition \ref{Def:MomentGeneratingFunction} takes the form
\begin{equation*} \Lambda_{\nu{^*}}(\theta) = \nu^*\{\ppp_A\}\ln \int_{r\in[0,1]} e^{\theta(1-r)}\wprho_A(\nu(\XX),dr) + \nu^*\{\ppp_B\}\ln \int_{r\in[0,1]} e^{\theta(1-r)}\wprho_B(\nu(\XX),dr) \end{equation*}
 for all $\theta\in \R$. The logarithmic moment generating function of $L_N$ of \eqref{E:lossapprox} is approximately $ N\Lambda_{\nu^*}(\theta/N)$.
 Thus we should have that
 \begin{equation*} \BP_N\lb L_N\in d\ell^*\big| \nu_N\approx \nu^*\rb
 \asymp \exp\left[-N \Lambda^*_{\nu^{*}}(\ell)\right] \end{equation*}
 We should then get the large deviations principle for $Z_N$ by combining this, \eqref{E:optent}, and \eqref{E:combine}, i.e.,
\begin{equation*}\ \BP_N\lb Z_N\in dz^*\rb \asymp \exp\left[-N [\Lambda^*_{\nu^{*}}(\ell)+H(\nu^{*})]\right]. \end{equation*}
\esquare
\end{example}

Examples \ref{Ex:onetypeLDP} and \ref{Ex:twotypesLDP} motivate our main result.
 \begin{theorem}[Main]\label{T:Main} Let Assumption \ref{A:LimitExists} hold and recall Definitions \ref{E:LFT} and \ref{Def:MomentGeneratingFunction}. Then, we have that $(Z_N)_{N\in \N}$ has a large deviations principle with rate function
 \begin{equation*} I(z) = \Lambda^*_\nu(\ell) + H(\nu) \end{equation*}
for all $z=(\ell,\nu)\in [0,1]\times \SubProb$.   Secondly $(L_N)_{N\in \N}$ has a large deviations principle with rate function
 \begin{equation}\label{E:varprob} I'(\ell) = \inf_{\nu\in \SubProb}I(\ell,\nu) \end{equation}
 for all $\ell\in [0,1]$.
\esquare\end{theorem}
\begin{proof} Combine together Propositions \ref{P:LSC}, \ref{P:LowerBound}, and \ref{P:Upper}.  This gives the large deviations
principle for $(Z_N)$.  The large deviations principle for $(L_N)_{N\in \N}$ follows from the contraction
principle and the continuity of the map $\nu\mapsto \nu(\XX)$.\end{proof}

Next we mention a useful representation formula for $H(\nu)$ that is defined in Definition \ref{E:LFT}. Its proof will be given in Section \ref{S:AuxiliaryLemmas}.  Define
\begin{equation*} \lambda_p(\theta) \Def \ln \left(pe^\theta+1-p\right)\end{equation*}
for all $\theta\in \R$ and $p\in [0,1]$.  Then $\lambda_p$ and $\hbar_p$ are convex duals; i.e.,
\begin{equation}\label{Eq:Duals_lambda_h}
\begin{aligned} \hbar_p(x) &= \sup_{\theta\in \R}\lb \theta x-\lambda_p(\theta)\rb \qquad x\in \R\\
\lambda_p(\theta) &= \sup_{x\in \R}\lb \theta x- \hbar_p(x)\rb \qquad \theta\in \R \end{aligned}
\end{equation}
and we have:
\begin{lemma}\label{L:Hrep} Let Assumption \ref{A:LimitExists} hold. Then, we have that
\begin{equation}\label{E:Hrep} H(\nu) = \sup_{\phi\in C(\XX)}\lb \int_{\ppp\in \XX}\phi(\ppp)\nu(dp) - \int_{\ppp=(p,\wprho)\in \XX}\lambda_p(\phi(\ppp)) \UU(d\ppp)\rb \end{equation}
for all $\nu\in \SubProb$.
\esquare
\end{lemma}

One way to interpret Theorem \ref{T:Main} is that the rate functions $I$ and $I'$ give the correct way to find the ``minimum-energy" configurations for atypically large
losses to occur.  In general, variational problems involving measures can be computationally difficult, so Section \ref{S:AlternateExpression} addresses some
computational issues.  In particular, we find an alternate expression which takes advantage of the specific structure of our problem.
Define
\begin{equation}\label{Eq:Duals_M_I}
\begin{aligned}
M_\wprho(\theta,D) &\Def \ln \int_{r\in [0,1]}e^{\theta (1-r)}\wprho(D,dr) \qquad \theta\in \R\\
I_\wprho(x,D) &\Def \sup_{\theta\in \R}\lb \theta x- M_\wprho(\theta,D)\rb\qquad x\in \R
\end{aligned}
\end{equation}
for all $D\in [0,1]$. Observe that we can write
\begin{equation}\label{E:LambdaI}  \Lambda_\nu(\theta) = \int_{\ppp=(p,\wprho)\in \XX} M_{\wprho}(\theta,\nu(\XX)) \nu(d\ppp) \end{equation}
for all $\nu\in \SubProb$ and $\theta\in \R$.
Define next $\BB \Def B(\XX;[0,1])$.
\begin{theorem}\label{T:AlternativeExpression} Let Assumption \ref{A:LimitExists} hold. For $\ell\in [0,1]$ and $\UU\in \Pspace(\XX)$, set
\begin{equation} \label{E:IDef}\begin{aligned} J'(\ell) &\Def \inf_{\Phi,\, \Psi \in \BB}\lb \int_{\ppp=(p,\wprho)\in \XX} \lb \Phi(\ppp)
I_\wprho\left(\Psi(\ppp),\int_{\ppp\in \XX}\Phi(\ppp)\UU(d\ppp)\right)+ \hbar_p(\Phi(\ppp))
\rb \UU(d\ppp):\right.\\
&\qquad \qquad \qquad \left.\int_{\ppp \in \XX} \Phi(\ppp)\Psi(\ppp)\UU(d\ppp)=\ell\rb.
\end{aligned}
\end{equation}
We have that $I'(\ell)=J'(\ell)$ for all $\ell\in [0,1]$.  An alternate representation of $J'$ is
\begin{equation} \label{E:IIDef}\begin{aligned} J''(\ell) &\Def \inf_{D\in [\ell,1]}\inf_{\Phi\in \BB}\inf_{\Psi\in \BB}\lb \int_{\ppp=(p,\wprho)\in \XX} \lb \Phi(\ppp)
I_{\wprho}\left(\Psi(\ppp),D\right)+ \hbar_p(\Phi(\ppp))
\rb \UU(d\ppp):\right.\\
&\qquad \qquad \qquad \left.\int_{\ppp\in \XX} \Phi(\ppp)\Psi(\ppp)\UU(d\ppp)=\ell,\, \int_{\ppp\in \XX} \Phi(\ppp)\UU(d\ppp)=D\rb.
\end{aligned}\end{equation}
\esquare
\end{theorem}
\noindent The proof of this is given in Section \ref{S:AlternateExpression}.  The point of the second representation \eqref{E:IIDef} is that the innermost minimization
problem (the one with $\Phi$ and $D$ fixed) involves linear constraints. Namely, that $\Phi$ takes values in $[0,1]$
and that two integrals of $\Phi$ take specific values.  This will be useful in some of our numerical studies in the next section.

\section{Examples and Discussion}\label{S:ExamplesDiscussion}
In Subsection \ref{SS:NumericalExamples} we present some  numerical examples. These examples showcase some of the possibilities and some of the implications
of the dependence of the recovery distribution on the defaults. We conclude this section with Subsection \ref{SS:Conclusions}, where we summarize our conclusions.

\subsection{Numerical Examples}\label{SS:NumericalExamples}
Let's see what our calculations look like in some specific cases. To
focus on the effects of recovery, let's assume a common (relatively high) probability
of default of $8\%$; i.e., $p^{N,n}=0.08$ for all $N\in \N$ and $n\in
\{1,2,\dots N\}$.

In the first group of examples (Cases 1-4), we will consider four specific cases, one with
fixed recovery rate, two homogeneous pools  with stochastic recoveries and one heterogeneous pool with stochastic recovery. For comparison purposes,
the law of large numbers average loss in the pool will be the same in all cases and equal to $\bar{L}=0.064$. We would like to remark the following.
In a more realistic scenario one would
expect to have pools with combinations of high-rated and low-rated counter-parties that have different default probabilities. Our formulation allows for such a scenario, but
since we want to focus on the effects of recovery we assume a common probability of default in all of our examples. Moreover, recall that we have assumed as a reference point
that all bonds have par value of $\$1$. We will see that the
tails (the large deviations principle rate functions) are
significantly different.  Although our theory has primarily focused on the rate function in the large deviations principle for $(L_N)_{N\in \N}$, the solution of the variational
problem \eqref{E:varprob}
(or equivalently \eqref{E:IDef} or \eqref{E:IIDef}) gives useful information.

In the second group of examples (Cases 5-6), we consider two heterogeneous pools, one with recovery whose distribution depends on the default rate (Case 5) and one with recovery whose distribution
does not depend on the default rate (Case 6). As expected,  the dependence on the default rate affects the tails.

In Cases 2-5 we consider stochastic recovery rates.  We want to consider the case that the recovery is in an appropriate sense negatively correlated with the defaults; i.e.,
 that more defaults imply less recovery.  In an economy that experiences recession, recovery rates tend to decrease just as defaults tend to increase.
This property is also  a documented empirical observation, e.g., see \cite{Das-Hanouna}, \cite{ABRS05} and the references therein.

In \textbf{Case 1}, let's assume that the recovery rate is fixed at $20\%$; i.e., $\wprho^{N,n}=\delta_{0.2}$ for all $N\in \N$ and $n\in \{1,2\dots N\}$.
 The LDP here is essentially a straightforward application of Cramer's Theorem. Let's see how our general formulation of Theorem \ref{T:Main} covers this as a special case.
In this case $\ppp= (0.08,\delta_{0.2})$ and $\UU= \delta_{\ppp}$.
Also we have
\begin{equation*}
H(\nu) = \begin{cases} \hbar_{0.08}(\nu(\XX)) &\text{if $\nu=\nu(\XX)\delta_{\ppp}$} \\
\infty &\text{else}\end{cases}
\end{equation*}
and
\begin{equation*}
\Lambda_\nu(\theta) =  0.8\theta \nu(\XX)
\quad\text{and}\quad
\Lambda^*_\nu(\ell) = \begin{cases} 0 &\text{if $\ell=0.8\nu(\XX)$} \\
\infty &\text{else}\end{cases}\end{equation*}
Collecting things together, we have that $(Z_N)_{N\in \N}$ and $(L_N)_{N\in \N}$ are governed, respectively, by the rate functions
\begin{align*} I_1(\ell,\nu) &= \begin{cases} \hbar_{0.08}\left(\frac{\ell}{0.8}\right) &\text{if $\nu= \frac{\ell}{0.8}\delta_{\ppp}$} \\
\infty &\text{else}\end{cases} \\
I_1'(\ell) &= \hbar_{0.08}\left(\frac{\ell}{0.8}\right) \end{align*}
We note that $I'_1(\ell)$ is finite only if $0\le \ell\le 0.8$.

In \textbf{Case 2} we consider a homogeneous pool with the recovery rate following a beta distribution. Essentially, this is a special case of Examples \ref{Ex:onetype} and \ref{Ex:onetypeLDP}. For $\beta>0$, define
\begin{equation*}\mu_\beta(A) \Def \beta\int_{r\in A}(1-r)^{\beta-1}dr, \qquad A\in \Borel[0,1]. \end{equation*}
this is the law of the beta distribution with parameters $1$ and $\beta$.
As $\beta$ increases, the amount of mass near $1$ decreases.  We also have that
\begin{equation*} \int_{r\in [0,1]}r \mu_\beta(dr) = \frac{1}{1+\beta} \end{equation*}
for all $\beta>0$ (as $\beta$ increases, the mean of $\mu_\beta$ decreases). This will allow a number of explicit formulae for the expected recovery (given the default rate).

We here assume that the recovery rate has a beta
distribution whose parameters depend linearly and monotonically on the empirical
default rate. Namely, if the default rate is $D$, then the
recoveries will all have common beta distribution with parameters
$1$ and
\begin{equation*}  \beta=f_{\aff}(D) \Def \frac{1}{0.2-0.1(D-0.08)}-1;\end{equation*}
note that $f:[0,1]\to \R_+$.  Define

\begin{align*} \wprho_{\aff}(D,\cdot) &=
\mu_{f_{\aff}(D)} \text{ for all }D\in [0,1]\\
\ppp_{\aff}&= (0.08,\wprho_{\aff})
\end{align*}
Then $\UU = \delta_{\ppp_{\aff}}$.

This choice of  $f_{\aff}$ results in a conditional expected recovery which is affine in $D$; i.e.,
\begin{equation*}
\int_{r\in [0,1]}r \wprho_{\aff}(D,dr) = 0.2-0.1(D-0.08).
\end{equation*}
Observe that the expected recovery of a single name in the homogeneous pool is decreasing in $D$ and always in  $[0,1]$.
According to \eqref{E:typdefrate},  the law of large numbers average loss is
\begin{equation*} \bar L = 0.08\times \left(1-\frac{1}{1+f_{\aff}(0.08)}\right) = 0.08\times \left(1-0.2\right) =0.064.\end{equation*}

The rate functions in this case are somewhat similar to those in Case 1.   Again we have that
\begin{equation*} H(\nu) = \begin{cases} \hbar_{0.08}(\nu(\XX)) &\text{if $\nu=\nu(\XX)\delta_{\ppp_{\aff}}$} \\
\infty. &\text{else}\end{cases} \end{equation*}
For each $\beta>0$, define
\begin{align*} \check \Lambda_\beta(\theta) &\Def \ln \int_{r\in [0,1]}e^{\theta(1-r)}\mu_\beta(dr) \qquad \theta\in \R \\
\check \Lambda^*_\beta(\ell) &\Def \sup_{\theta\in \R}\lb \theta \ell -\check \Lambda_\beta(\theta)\rb \qquad \ell\in \R. \end{align*}
Then if $\nu=\nu(\XX)\delta_{\ppp_{\aff}}$ where $\nu(\XX)>0$ we have from Definition \ref{Def:MomentGeneratingFunction},
\begin{align*} \Lambda_\nu(\theta) &= \nu(\XX)\check \Lambda_{f_{\aff}(\nu(\XX))}(\theta) \qquad \theta\in \R\\
\Lambda^*_\nu(\ell) &= \nu(\XX)\check \Lambda^*_{f_{\aff}(\nu(\XX))}\left(\frac{\ell}{\nu(\XX)}\right) \qquad \ell\in \R. \end{align*}
Here $(Z_N)_{N\in \N}$ and $(L_N)_{N\in \N}$ are governed, respectively, by the rate functions.
\begin{equation} \label{E:Iprimes} \begin{aligned} I_2(\ell,\nu) &= \begin{cases} \hbar_{0.08}(\nu(\XX)) +\nu(\XX)\check \Lambda^*_{f_{\aff}(\nu(\XX))}\left(\frac{\ell}{\nu(\XX)}\right) &\text{if $\nu=\nu(\XX)\delta_{\ppp_{\aff}}$} \\
\infty &\text{else}\end{cases} \\
I_2'(\ell) &= \inf_{D\in (0,1]}\lb \hbar_{0.08}(D) + D \check \Lambda^*_{f_{\aff}(D)}\left(\frac{\ell}{D}\right)\rb.  \end{aligned}\end{equation}

In \textbf{Case 3}, we replace $f_\aff$ of Case 2 with one that results in a conditional expected recovery which is quadratic in $D$;
this allows us some insight into the effects of convexity in the conditional expected recovery.  We set
\begin{equation*} f_{\qq}(D)=\frac{1}{0.2-0.1(D-0.08)-0.1(D-0.08)^2}-1.\end{equation*}
Again we have that $f_{\qq}:[0,1]\to \R_+$, and we set $\wprho_{\qq}(D,\cdot) = \mu_{f_{\qq}(D)}$, $\ppp_{\qq}= (0.08,\wprho_{\qq})$, and have that $\UU=\delta_{\ppp_{\qq}}$
Here we have that
\begin{equation*}
\int_{r\in [0,1]}r \wprho_{\qq}(D,dr) = 0.2-0.1(D-0.08)-0.1(D-0.08)^2.
\end{equation*}
Observe that the expected recovery of a single name in the homogeneous pool is decreasing in $D$ and always in $[0,1]$.

We again get that $\bar L=0.064$.  The corresponding rate function is $I'_3$.

\textbf{Case 4} involves the beta distribution again. Here, however, we now consider a heterogeneous pool of two types
(Examples \ref{Ex:twotypes} and \ref{Ex:twotypesLDP}). We concentrate on the effect of the heterogeneity in the recovery distribution, so as in
the previous cases all bonds will have default probability of $8\%$.
For all $D\in [0,1]$, every third bond will have recovery distribution
governed by  $\wprho_{\aff}$ and the remaining bonds will have recovery distribution governed by $\wprho_{\qq}$.  We thus have that $\UU = \tfrac13 \delta_{\ppp_{\aff}} + \tfrac23 \delta_{\ppp_{\qq}}$.
It is easy to see that again the law of large numbers average loss is:
\begin{equation*} \bar L = 0.08\times\frac{1}{3}\times \left(1-\frac{1}{1+f_{\aff}(0.08)}\right)+ 0.08\times \frac{2}{3}\times \left(1-\frac{1}{1+f_{\qq}(0.08)}\right)=0.064.
\end{equation*}

For notational convenience and in order to illustrate the usage of Theorem \ref{T:AlternativeExpression} we  use the alternative representation \eqref{E:IIDef}.  If $\wprho(D,\cdot)=\mu_{f(D)}$
for some $f\in C([0,1];\R_+)$, then for all $D\in [0,1]$, we have that $M_{\wprho}(\theta,D) = \check \Lambda_{f(D)}(\theta)$ for all $\theta\in \R$ and $I_\wprho(\ell,D) = \check \Lambda^*_{f(D)}(\ell)$
for all $\ell\in [0,1]$.   Since the support of $\UU$ is exactly $\{\ppp_{\aff},\ppp_{\qq}\}$, we can consider $\Phi$ and $\Psi$ in $\BB=B(\XX;[0,1])$ of the form
\begin{equation*} \Phi= \phi_A\chi_{\{\ppp_{\aff}\}} + \phi_B\chi_{\{\ppp_{\qq}\}}  \qquad \text{and}\qquad  \Psi= \psi_A\chi_{\{\ppp_{\aff}\}}  + \psi_B\chi_{\{\ppp_{\qq}\}} . \end{equation*}
Thus \eqref{E:IIDef} becomes
\begin{equation} \begin{aligned}\label{E:Ihet} I'_4(\ell) &= \inf_{D\in[\ell,1]}\inf_{\psi_A,\psi_B\in [0,1]}\inf_{\substack{\phi_A,\phi_B\in [0,1]\\  \phi_A\psi_A/3+2\phi_A\psi_B/3=\ell \\ \phi_A/3+2\phi_B/3=D}}\lb \frac13 \phi_A I_{\wprho_{\aff}}(\psi_A,D)+ \frac23 \phi_B I_{\wprho_{\qq}}(\psi_B,D)\right.\\
&\qquad\qquad \left.+\frac{1}{3} \hbar_{p_A}(\phi_A) +\frac{2}{3} \hbar_{p_B}(\phi_B)\rb \\
 &= \inf_{D\in[\ell,1]}\inf_{\psi_A,\psi_B\in [0,1]}\inf_{\substack{\phi_A,\, \phi_B\in [0,1]\\  \phi_A\psi_A/3+2\phi_A\psi_B/3=\ell \\ \phi_A/3+2\phi_B/3=D}}\lb \frac13 \phi_A \check \Lambda^*_{f_\aff(D)}(\psi_A)+ \frac23 \phi_B \check \Lambda^*_{f_\qq(D)}(\psi_B)\right.\\
 &\qquad\qquad \left.+\frac{1}{3} \hbar_{p_A}(\phi_A) +\frac{2}{3} \hbar_{p_B}(\phi_B)\rb \end{aligned}\end{equation}


In Figure \ref{F:Figure1}, we plot the rate functions $I'_1$,
$I'_2$, $I'_3$ and $I'_4$.  We use a Monte Carlo procedure to compute $\check \Lambda$ and $\check \Lambda^*$.
As expected, all action functions are nonnegative and zero at the (common) law of large numbers average loss of $\bar L=0.064$.
\begin{figure}[ht]
\begin{center}
\includegraphics[scale=0.4, width=9 cm, height=6 cm]{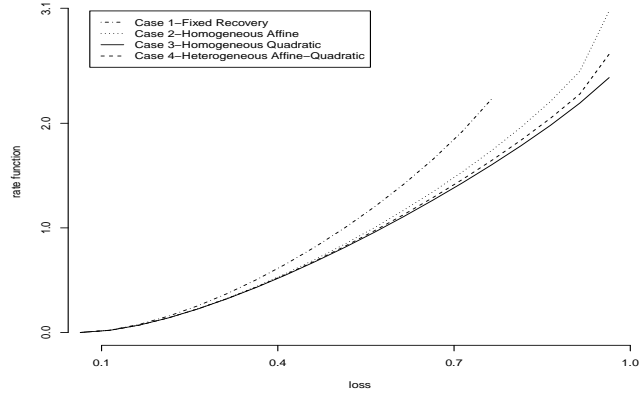}
\caption{Rate functions for fixed recovery, for the homogeneous cases and the heterogeneous case.}
\end{center}\label{F:Figure1}
\end{figure}
We observe that 
$I'_3\leq I'_4 \leq I'_2\leq I'_1$.
 In particular, the heterogeneous case,
which is a mixture of an affine conditional expected recovery and a quadratic conditional expected recovery, is in between the two homogeneous cases
 (cases 2 and 3).  Of course, we should not be surprised that the rate function in Case 1 is larger than that in Cases 2 through 4,
there are in general many more configurations which lead to a given overall loss rate. An interesting observation that seems to be suggested by the examples considered here
is that more convexity results in smaller values for the rate function.

The discussion from now on will be formal. However, we believe that it offers some useful insights in the effect of default rates on the average recovery of large pools.

Another useful insight which we can numerically extract is the ``preferred'' way which losses stem from defaults versus recovery. For each $\ell\in [0,1]$, let $D_*(\ell)$
be the minimizer\footnote{Case 1 is of course degenerate this sense; for a given loss rate $\ell$, the default rate must be very close to $\ell/0.8$.} in the expression
\eqref{E:Iprimes} for $I'_2$, $I'_3$ 
or alternately the expression \eqref{E:Ihet} for $I'_4$. $D_*(\ell)$ can be interpreted  as the ``most likely`` default rate in the pool given that $L_{N}\approx \ell$.
 We make the following assumption.
\begin{assumption}\label{A:MinimizerUniqueness}
Assume that for each $\ell\in [0,1]$ the aforementioned minimizer  $D_*(\ell)$ exists and is unique.
\end{assumption}
Numerically finding $D_*(\ell)$ in the examples considered here, verifies that Assumption \ref{A:MinimizerUniqueness} holds in these cases.
For $\ell>\bar L$ and $\delta>0$, the Gibbs conditioning principle \cite[Section 7.3]{MR1619036} implies that we should have that
\begin{equation}\label{E:Gibbs} \lim_{N\to \infty}\BP_N\lb |D_N-D_*(\ell)|\ge \delta\big| L_N\ge \ell\rb = 0. \end{equation}
In other words, conditional on the pool suffering losses exceeding rate $\ell$, the default rate should converge to $D_*(\ell)$.
Using this information, we can then say something about how the average recovery of the pool is related to the optimal default rate $D_*(\ell)$.
Motivated by (\ref{E:typdefrate}) we write that
\begin{equation*}
\text{Average Loss=Default $\times$ (1-Recovery)}\Rightarrow \text{Recovery=1-$\frac{\text{Average Loss}}{\text{Default}}$}
\end{equation*}
to find an effective recovery rate in terms of the loss rate and the default rate.  This formulation quantifies the fact that losses are due to \emph{both} default and recovery.  For atypically large
losses in a large pool of credit assets, we should combine this with the Gibbs conditioning calculation of \eqref{E:Gibbs}.  Namely, let's define
\begin{equation*}
\IRec(\ell)=1-\frac{\ell}{D_*(\ell)}.
\end{equation*}
We can also formalize the dependence of recovery on default by letting $\IRec^*:[0,1]\to [0,1]$ be such that
\begin{equation}
\IRec^*(D_*(\ell))=\IRec(\ell) \text{ for all }\ell\in [0,1].\label{Eq:OptimalAverageRecovery}
\end{equation}
We refer to $\IRec^*(\cdot)$ as the \emph{effective average recovery} of the pool.

Figure 5.2 is a plot of $\IRec^*$ for the cases which we are studying. As it was expected, we observe that,
for Cases $2,3$ and $4$, the optimal average recovery of the pool and the optimal defaults are negatively correlated. Of course, this is consistent with the corresponding negative
correlation of the individual recovery and default rate that is embedded in the choices of the individual recovery distribution. Notice however that  $\IRec^*$
 is a ''global" quantity in that it represents the effective average recovery in the pool.
 Another interesting observation is that the graph of the heterogeneous case which is a mixture of an affine and a quadratic conditional expected recovery is between
 the graph of the homogeneous cases 2 and 3 which treat the affine and quadratic case respectively.
\begin{figure}[ht]
\begin{center}
\includegraphics[scale=0.4, width=9 cm, height=6 cm]{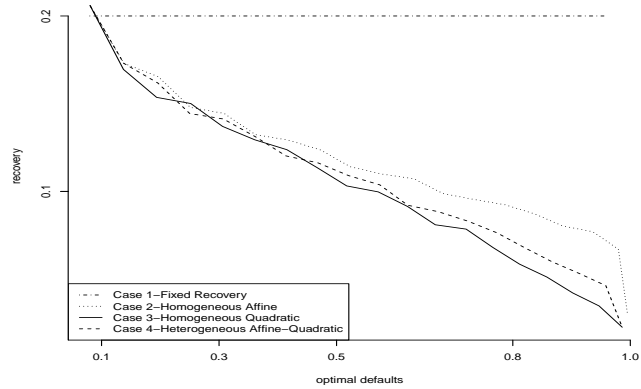}
\caption{Effective average recovery of the pool versus optimal defaults.}
\end{center}\label{F:Figure2}
\end{figure}

We conclude this subsection with a comparison of a heterogeneous portfolio whose distribution of recovery depends on the default rate with one that does not.

\textbf{Case 5.} Consider the setup of Case 4 with the only difference being the definitions of $f_{\aff}(D)$ and $f_{\qq}(D)$. In particular, we consider
\begin{align*}
 f_{\aff}(D) &= \frac{1}{0.1-0.05(D-0.08)}-1, \\
f_{\qq}(D)&=\frac{1}{0.25-0.1(D-0.08)-0.1(D-0.08)^2}-1.
\end{align*}
The corresponding rate function is $I'_5$.

\textbf{Case 6.} Again, consider the setup of Case 4 with the only difference being the definitions of $f_{\aff}(D)$ and $f_{\qq}(D)$. In particular, we consider
\begin{align*}
 f_{\aff} &= \frac{1}{0.1}-1,\\
f_{\qq}&=\frac{1}{0.25}-1.
\end{align*}
Notice that in this case the distribution of the recovery is independent of the default rate $D$. The corresponding rate function is $I'_6$.

In both cases the law of large numbers average loss is the same as before. In Figure 5.3, we plot (a): on the left, the rate functions $I'_5$
and $I'_6$ and (b): on the right, $\IRec^*$ for cases 5 and 6.

\begin{figure}[ht]
\begin{center}
\includegraphics[scale=0.4, width=7 cm, height=5 cm]{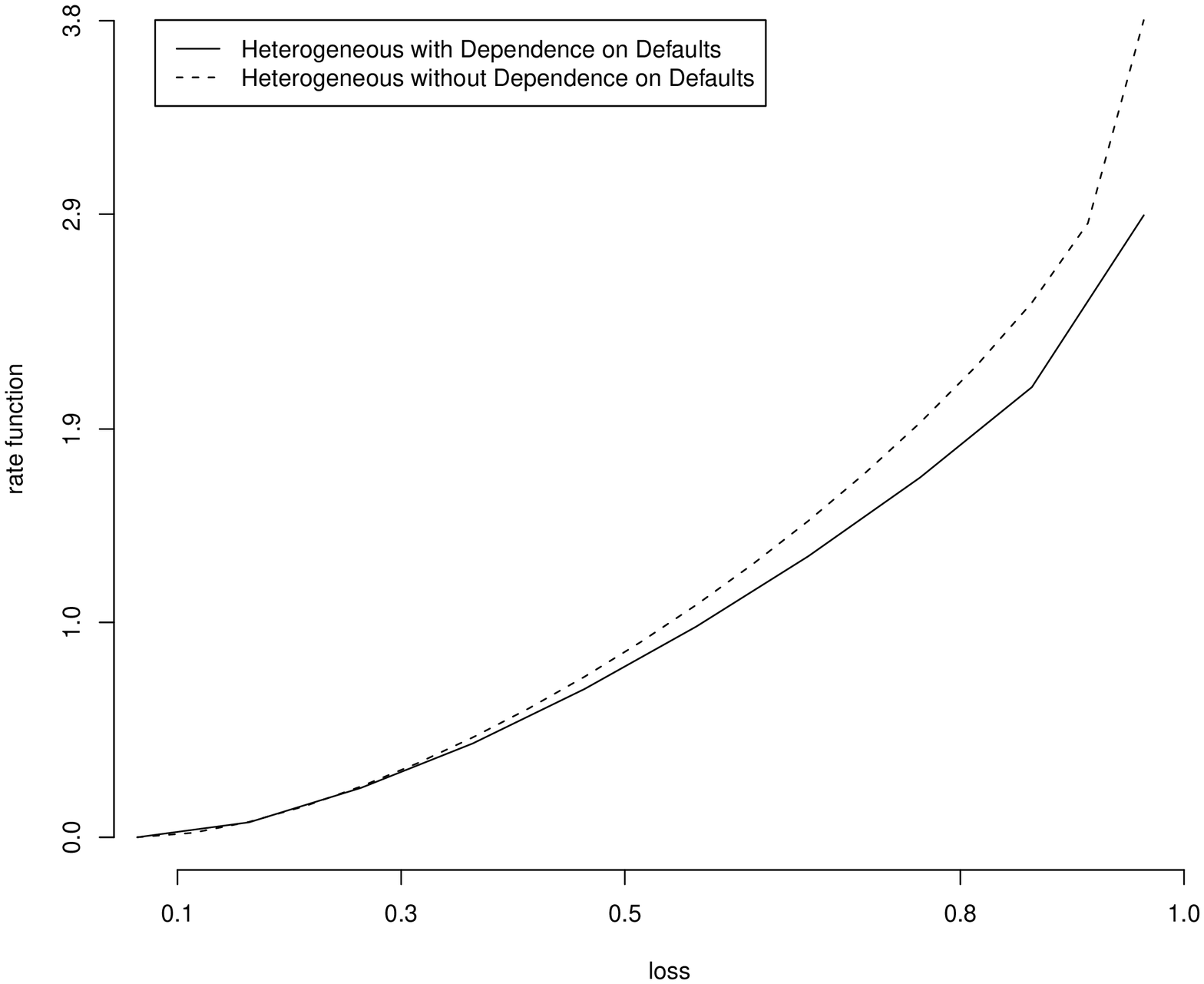}
\hspace{0.1cm}
\includegraphics[scale=0.4, width=7 cm, height=5 cm]{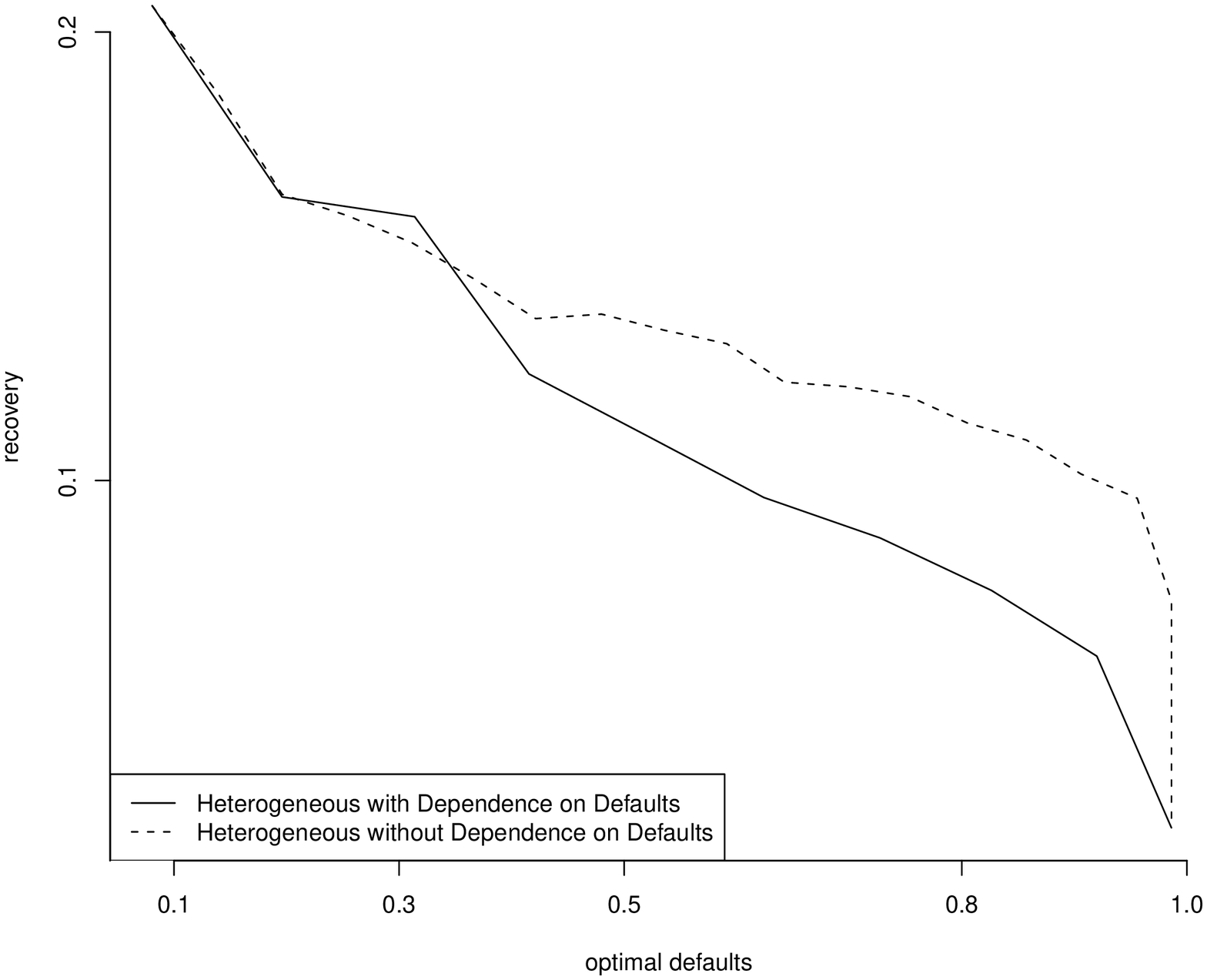}
\caption{Rate functions and $\IRec^*$ when recovery distribution depends on defaults and when it does not.}
\end{center}\label{F:Figure3}
\end{figure}

As it is indicated by the left figure, one of the effects of the dependence of the distribution of the recovery on the defaults is to decrease the values of the rate function, i.e.,
$I'_{5}\leq I'_{6}$. The effective average recovery, $\IRec^*$, is in both cases negatively correlated with the optimal default rate $D_*$. Notice however that for large overall default rates
the effective average recovery of the heterogeneous pool with dependence on defaults (Case 5) is less than the effective average recovery
of the heterogeneous pool without dependence on defaults (Case 6). This is consistent with our intuition, namely that dependence of the recoveries on defaults should affect the recovery
of the pool and in particular that more defaults should decrease recovery.

\subsection{Conclusions and Discussion}\label{SS:Conclusions}
In this subsection we summarize our findings and pose some question  that would be interesting to study.

\begin{itemize}
\item{Assuming that the recoveries for the defaulted assets depend upon the number of defaults in a fairly general way we have characterized the typical (Lemma \ref{L:LLN}) and atypical (Theorems \ref{T:Main} and \ref{T:AlternativeExpression}) behavior of the loss rates in the pool. This allows for consideration of the
     case when recovery rates are affected by the number of defaults.
We prove in Theorem \ref{T:Main} the large deviations principle  for the joint family of random variables $(L_{N}, \nu_{N})$
where $L_N$ is the average loss and $\nu_N$ is the  empirical measure on type-space determined
by the names which default. Then, the large deviations principle for $L_{N}$ follows by the contraction principle.  Furthermore, we derive in Theorem \ref{T:AlternativeExpression} various equivalent  representations for the rate function which give some more insight to the favored way to rearrange recoveries and losses among the different types and also ease its numerical computation.}
\item{Moreover, as we have demonstrated in Subsection \ref{SS:NumericalExamples}, the rate functions can be calculated numerically for both homogeneous and heterogeneous cases.
 The rate function determines the main term in the logarithmic asymptotics of the probability that the average loss in the pool is, let's say, bigger than a specific level. In particular,
 it determines the main term of the tail distribution of total losses on a portfolio consisting of many positions. See also \cite{MR2022976,Gordy02} for a related discussion. Also note that
in the examples considered here, more convexity resulted in smaller rate functions.}
  \item{Our formulation allows to extract some information regarding the optimal default rate in the pool for a given level of average loss in the pool. Moreover, a useful insight into the fact that losses are due to \emph{both} default and recovery is perhaps given by the \emph{effective average recovery} as defined by relation
(\ref{Eq:OptimalAverageRecovery}). 
}
\item{As it is indicated by the comparison of Cases 5 and 6, the effect of the dependence of the distribution of the recovery on the default rate is (a): to reduce the rate function and
(b): to reduce the effective average recovery of the pool especially when the optimal default rate gets larger.}
\end{itemize}

Some interesting questions that naturally arise are  below. These questions will be pursued elsewhere.
\begin{itemize}
\item In this paper we deal with logarithmic asymptotics. It would be interesting to study the exact asymtpotics and characterize the prefactor that appears in the asymptotic
 expansion  of the loss distribution as the pool gets larger. Similar questions have also been studied in \cite{MR2022976,leijdekker-2009} under various assumptions.
\item A question that is in particular relevant for financial applications is to study measures such as Value-at-Risk (VaR) and expected shortfall (ES). VaR at level $a\in(0,1)$ is the a-quantile
of the loss distribution and expected shortfall at level a is defined as $ES_{a}=\BE[L|L\geq VaR_{a}(L)]$. It would be interesting to study the asymptotic behavior of $VaR_{a}(L_{N})$ and
$ES_{a}(L_{N})$, characterize their limits as the pool gets larger and study the corresponding implications of the dependence of the recoveries on defaults. Under certain conditions similar
questions have also been investigated in \cite{Gordy03, FM02,FM03}.
\end{itemize}

\section{Compactness of Level Sets}\label{S:lsc}
The first part of the large deviations claim is that the level sets
of $I$ are compact.  The proof follows along fairly standard lines.

First, however, we need to topologize $\SubProb$. This is done in the usual way. In particular, fix a point $\pt$ that is not in $\XX$ and define $\XX^+\Def \XX\cup \{\pt\}$.
Give $\XX^+$ the standard topology; open sets are those which are open subsets of $\XX$ (with its original topology) or complements in $\XX^+$ of closed subsets
of $\XX$ (again, in the original topology of $\XX$).  Define a bijection $\iota$ from $\SubProb$ to $\Pspace(\XX^+)$ by setting
\begin{equation*} (\iota \nu)(A) \Def \nu(A\cap \XX) + \left(1-\nu(\XX)\right)\delta_{\pt}(A) \end{equation*}
for all $A\in \Borel(\XX^+)$.  The point $\pt$ is introduced because $\nu$  is a subprobability measure. The topology of $\SubProb$ is the pullback of the topology of $\Pspace(\XX^+)$ and the metric on $\SubProb$ is that given by
requiring $\iota$ to be an isometry.

Since $\XX$ is Polish, so is $\XX^+$, and thus $\Pspace(\XX^+)$ is Polish, and thus $\SubProb$ is a Polish space.

\begin{proposition}\label{P:LSC} For each $s\ge 0$, the set
\begin{equation*} \Phi(s) \Def \lb z\in [0,1]\times \SubProb: I(z)\le s\rb \end{equation*}
is a compact subset of $[0,1]\times \SubProb$.
\esquare
\end{proposition}
\begin{proof} We first claim that $\Phi(s)$ is contained in a compact subset of $[0,1]\times \SubProb$.  Since $[0,1]$ is already compact, it suffices to show that
$\Phi_M(s)\Def \lb \nu\in \SubProb: H(\nu)\le s\rb $ is a compact subset of $\SubProb$.  If $\nu\in \Phi_M(s)$, then $\nu \ll \UU$ and, since $\hbar_p(x)=\infty$ for $x>1$, we have that
\begin{equation*} \UU\lb \ppp\in \XX:\, \frac{d\nu}{d\UU}(\ppp)>1\rb =0, \end{equation*}
so for any $A\in \Borel(\XX)$, $\nu(A)\le \UU(A)$.  Since $\UU$ itself is tight (it is a probability measure on a Polish space), $\Phi_M(s)$ is tight; for every $\eps>0$,
there is a $K_\eps\subset \subset \XX$ such that $\nu(\XX\setminus K_\eps)<\eps$ for all $\nu\in \Phi_M(s)$.  We claim
that thus $\iota(\Phi_M(s))$ is also tight.  Indeed, fix $\eps>0$.  Letting $\iota_\circ:\XX\to \XX^+$ be the inclusion map, we have that $\iota_\circ$ is continuous, and thus $\iota_\circ(K_\eps)$ is
compact.  Since singletons are also compact, $K^*\Def \iota_\circ(K_\eps)\cup \{\pt\}$ is a compact subset of $\XX^+$.  For every $\nu\in \Phi_M(s)$,
$(\iota \nu)(\XX^+ \setminus K^*) = \nu(\XX\setminus K_\eps)<\eps$, so indeed $\iota(\Phi_M(s))$ is tight.  Thus $\overline{\Phi_M(s)}\subset \subset \Pspace(\XX^+)$ and
hence
\begin{equation*} \Phi_M(s)\subset \iota^{-1}\overline{\iota(\Phi_M(s))}\subset \subset \SubProb \end{equation*}
the last claim following since $\iota$ is a homeomorphism.  Gathering things together, we have that $\Phi(s)$ is indeed contained in a compact subset of
$[0,1]\times \SubProb$.

We now want to show that $\Phi(s)$ is closed, or equivalently, that $([0,1]\times \SubProb) \setminus \Phi(s)$ is open.  Using Lemma \ref{L:Hrep}, we have that
\begin{multline*} ([0,1]\times \SubProb) \setminus \Phi(s) \\
= \bigcup_{\substack{\theta\in \R \\ \phi\in C(\XX)}}\lb (\ell,\nu)\in \SubProb:\, \theta \ell + \int_{\ppp\in \XX}\phi(\ppp)\nu(d\ppp) - \Lambda_\nu(\theta) > s + \int_{\ppp=(p,wp)\in \XX}\lambda_p(\phi(\ppp))\UU(d\ppp)\rb. \end{multline*}
For each $\theta\in \R$ and $\phi\in C(\XX)$, then map $(\ell,\nu)\mapsto \theta \ell + \int_{\ppp\in \XX}\phi(\ppp)\nu(d\ppp) - \Lambda_\nu(\theta)$
is continuous, so we have written $([0,1]\times \SubProb) \setminus \Phi(s)$ as a union of open sets.
\end{proof}

\section{Large Deviations Lower Bound}\label{S:Lower}
We next prove the large deviations lower bound.   As with most large deviations lower bounds, the idea is to find a measure
transformation under which the set of interest becomes ``typical''.  In this case, this measure transformation will come from a combination of
Cramer's theorem and Sanov's theorem.

First, we mention an auxiliary approximation result which will be useful in the proof. Its proof is in Section \ref{S:AuxiliaryLemmas}.
\begin{lemma}\label{L:approximation}  Fix $\nu\in \SubProb$ such that $H(\nu)<\infty$.  Then there is a sequence $(\nu_N)_{n\in \N}$ such that
\begin{equation*} \lim_{N\to \infty}\nu_N = \nu \qquad \text{and}\quad \lim_{N\to \infty}H(\nu_N)=H(\nu) \end{equation*}
\textup{(}and thus $\nu_N\ll \UU$ for all $N\in \N$\textup{)} and such that
\begin{equation*} \ppp\mapsto \frac{d\nu_N}{d\UU}(\ppp) \qquad \text{and}\qquad \ppp=(p,\wprho)\mapsto \chi_{(0,1)}(p)\hbar'_p\left(\frac{d\nu_N}{d\UU}(\ppp)\right)\end{equation*}
are both well-defined and in $C(\XX)$ for all $N\in \N$.\end{lemma}

We start with a simplified lower bound where the measure transformation in Cramer's theorem is fairly explicit.   For each $\nu\in \SubProb$, we make the usual definition \cite{MR1619036}[Appendix A] that
\begin{equation*} \dom \Lambda^*_\nu \Def \{\ell\in [0,1]: \Lambda^*_\nu(\ell)<\infty\}; \end{equation*}
this will of course be an interval; $\ridom \Lambda^*_\nu$ will be the relative interior of $\dom \Lambda^*_\nu$.

\begin{proposition}\label{P:lowerridom}  Fix an open subset $G$ of $[0,1]\times \SubProb$ and $z=(\ell,\nu)\in G$ such that $I(z)<\infty$ and $\ell \in \ridom \Lambda^*_\nu$.  Then
\begin{equation}\label{E:lowerridom} \varliminf_{N\to \infty}\frac{1}{N}\ln \BP_N\lb Z_N\in G\rb  \ge -I(z). \end{equation}
\esquare
\end{proposition}
\begin{proof}
The proof will require a number of tools. For presentation purposes we split the proof in three  steps. In Step 1,
we prove several auxiliary results and are then used to define and analyze the measure change that is used to prove the initial lower bound. In Step 2,
 we prove that under the measure defined by (\ref{Eq:InitalLowerBound3}) in Step 1, the recovery rates for the names that have defaulted are independent and that
the default probabilities are independent as well. In Step 3, we put things together in order to prove the initial lower bound (\ref{E:lowerridom}).

\textbf{Step 1.} Since $\ell\in \ridom \Lambda^*_\nu$, there is a $\theta\in \R$ such that
\begin{equation}\label{E:rid} \Lambda'_\nu(\theta)=\ell \qquad \text{and}\qquad \Lambda^*_\nu(\ell) = \theta \Lambda_\nu'(\theta)-\Lambda_\nu(\theta) \end{equation}
(see \cite{MR1619036}[Appendix A]).
Let's now fix a relaxation parameter $\eta>0$.  Then there is an $\eta_1\in (0,\eta)$ and an open neighborhood $\OO_1$ of $\nu$ such that
$(\ell-\eta_1,\ell+\eta_1)\times \OO_1 \subset G$.  Using the first equality of \eqref{E:rid}, we have that $(\Lambda'_\nu(\theta),\nu)= (\ell,\nu)\in (\ell-\eta_1,\ell+\eta_1)\times \OO_1$.
Since the maps $(\tilde \eta,\tilde \nu)\to (\Lambda'_{\tilde \nu}(\theta)+\tilde \eta,\tilde \nu)$ and $\tilde \nu\mapsto \Lambda_{\tilde \nu}(\theta)$ are continuous, there is an $\eta_2\in (0,1)$ and an open subset $\OO_2$ of $\SubProb$ such that
\begin{equation}\label{Eq:InitalLowerBound1}
\begin{aligned} \lb (\Lambda'_{\tilde \nu}(\theta)+\tilde \eta,\tilde \nu): \tilde \eta\in (0,\eta_2), \, \tilde \nu\in \OO_2\rb &\subset (\ell-\eta_1,\ell+\eta_1)\times \OO_1 \\
\left|\Lambda_{\tilde \nu}(\theta)-\Lambda_\nu(\theta)\right| < \eta \qquad \text{for $\tilde \nu\in \OO_2.$}
\end{aligned}
\end{equation}
We next want to use Lemma \ref{L:approximation}  to choose a particularly nice element of $\OO_2$.  Namely, Lemma \ref{L:approximation} ensures that there is
a $\nu^*\in \OO_2$ such that $\nu^*\ll \UU$ and such that both $\tfrac{d\nu^*}{d\UU}$ and
\begin{equation}\label{E:phidef} \phi(\ppp) \Def \chi_{(0,1)}(p)\hbar'_p\left(\frac{d\nu^*}{d\UU}(\ppp)\right) \qquad \ppp=(p,\wprho)\in \XX\end{equation}
are in $C(\XX)$ and such that $\left|H(\nu^*)-H(\nu)\right|<\eta$.
Let $\OO_3$ be an open subset of $\OO_2$ which contains $\nu^*$ and such that
\begin{equation*}\left|\int_{\ppp\in \XX}\phi(\ppp)\tilde \nu(d\ppp)-\int_{\ppp\in \XX}\phi(\ppp)\nu^*(d\ppp)\right|<\eta \end{equation*}
for all $\tilde \nu\in \OO_3$.

We can now proceed with our measure change.  For each $N\in \N$, define
\begin{equation}
A^{(N)}_1\Def \theta L_N- \Lambda_{\nu_N}(\theta)\qquad \text{and}\qquad A^{(N)}_2 \Def \frac{1}{N}\sum_{n=1}^N\Delta_n \phi(\ppp^{N,n}) - \lambda_{p^{N,n}}(\phi(\ppp^{N,n})) \label{Eq:InitalLowerBound2}
\end{equation}
Standard calculations show that
\begin{equation}\label{Eq:InitalLowerBound2a}
 \BE_N\left[\exp\left[NA^{(N)}_1\right]\bigg|\dilt\right]=1 \qquad \text{and}\qquad \BE_N\left[\exp\left[NA^{(N)}_2\right]\right]=1. \end{equation}
Define a new probability measure as
\begin{equation}
 \tilde \BP_N(A) \Def \BE_N\left[\chi_A\exp\left[N\lb A^{(N)}_1+A^{(N)}_2\rb \right]\right].\qquad A\in \Borel(E)\label{Eq:InitalLowerBound3} \end{equation}
This will be the desired measure change.

Define
\begin{equation}
 S_N \Def \lb \left|L_N-\Lambda'_{\nu_N}(\theta)\right|<\eta_2,\, \nu_N\in \OO_3\rb. \label{Eq:InitalLowerBound4}
 \end{equation}
On $S_N$,
\begin{equation*}(L_N,\nu_N)= \left(\Lambda'_{\nu_N}(\theta)+ \lb L_N-\Lambda'_{\nu_N}(\theta)\rb,\nu_N\right) \in (\ell-\eta_1,\ell+\eta_1)\times \OO_1\subset G \end{equation*}
so in fact $S_N\subset G$.
Thus
\begin{equation*} \BP_N\{Z_N\in G\} \ge  \tilde \BP_N\left[\chi_{S_N}\exp\left[-N\lb A^{(N)}_1 + A^{(N)}_2\rb\right]\right]. \end{equation*}
Let's also assume that $N$ is large enough that
\begin{equation*} \left| \int_{\ppp=(p,\wprho)\in \XX}\lambda_p(\phi(\ppp))\UU_N(d\ppp)- \int_{\ppp=(p,\wprho)\in \XX}\lambda_p(\phi(\ppp))\UU(d\ppp)\right|<\eta. \end{equation*}
Thus by (\ref{E:rid}),(\ref{Eq:InitalLowerBound1}),(\ref{Eq:InitalLowerBound2}) and (\ref{Eq:InitalLowerBound4}) we have that
\begin{equation*} A^{(N)}_1 \le \theta  \ell + |\theta| \eta_1 - \Lambda_\nu(\theta) + \eta = \Lambda^*_\nu(\ell) + \left(|\theta|+1\right)\eta.\end{equation*}

Next we prove that
\begin{equation}\label{E:hope} \hbar_p\left(\frac{d\nu^*}{d\UU}(\ppp)\right) = \frac{d\nu^*}{d\UU}(\ppp)\phi(\ppp)-\lambda_p(\phi(\ppp)) \end{equation}
for $\UU$-almost all $\ppp=(p,\wprho)\in \XX$.  This follows from standard convex analysis and the form \eqref{E:phidef} of $\phi$ when $p\in (0,1)$.  Since $H(\nu^*)<\infty$, \eqref{E:smallb}
implies that, except on a $\UU$-negligible set,
\begin{equation*} \hbar_p\left(\frac{d\nu^*}{d\UU}(\ppp)\right) = \hbar_p(p) = 0 = p\times 0 - \lambda_p(0) = \frac{d\nu^*}{d\UU}(\ppp)\phi(\ppp)-\lambda_p(\phi(\ppp)) \end{equation*}
if $\ppp=(p,\wprho)\in \XX$ is such that $p\in \{0,1\}$.  In other words, \eqref{E:hope} holds except on a $\UU$-negligible set.

Therefore,
\begin{multline}\label{Eq:InitalLowerBound5}
A^{(N)}_2 = \int_{\ppp\in \XX}\phi(\ppp)\nu_N(d\ppp) - \int_{\ppp=(p,\wprho)\in \XX}\lambda_p(\phi(\ppp))\UU_N(d\ppp)\\
\le  \int_{\ppp\in \XX}\phi(\ppp)\nu^*(d\ppp) - \int_{\ppp=(p,\wprho)\in \XX}\lambda_p(\phi(\ppp))\UU(d\ppp) + 2\eta\\
=  \int_{\ppp=(p,\wprho)\in \XX}\lb \phi(\ppp)\frac{d\nu^*}{d\UU}(\ppp)-\lambda_p(\phi(\ppp))\rb \UU(d\ppp) + 2\eta\\
= \int_{\ppp=(p,\wprho)\in \XX}\hbar_p\left(\frac{d\nu^*}{d\UU}(\ppp)\right)\UU(d\ppp) + 2\eta \le H(\nu) + 3\eta. \end{multline}
The last line of (\ref{Eq:InitalLowerBound5}) follows from (\ref{E:hope}).

Thus
\begin{equation}
 \BP_N\{Z_N\in G\} \ge \tilde \BP_N(S_N)\exp\left[-N\lb I(z) - \left(|\theta|+4\right)\eta\rb \right] \label{Eq:InitalLowerBound6}
 \end{equation}

\textbf{Step 2.}  Let's understand the law of $\{\ell_n\}_{1\le n\le N}$ under $\tilde \BP_N\{\cdot | \dilt\}$ defined by (\ref{Eq:InitalLowerBound3}). In particular we prove that under the measure defined by (\ref{Eq:InitalLowerBound3}) the recovery rates for the names that have defaulted are independent and that the default probabilities are independent as well.

 For any $\{\psi\}_{1\le n\le N}\subset \R$, we have that
 \begin{multline*} \BE_N\left[\exp\left[\sqrt{-1}\sum_{n=1}^N \psi_n\ell_n + N\theta L_N\right]\Bigg| \dilt\right]\\
 = \prod_{n=1}^N \lb \Delta_n \int_{r\in [0,1]}\exp\left[\left(\sqrt{-1}\psi_n + \theta\right)(1-r)\right]\wprho_{N,n}(\nu_N(\XX),dr) + 1-\Delta_n\rb. \end{multline*}
 Thus
 \begin{equation*}  \tilde \BE_N\left[\exp\left[\sqrt{-1}\sum_{n=1}^N \psi_n\ell_n\right]\Bigg| \dilt\right]
 = \prod_{n=1}^N \lb \Delta_n \int_{r\in [0,1]}\exp\left[\sqrt{-1}\psi_n(1-r)\right]\tilde \wprho_{N,n}(\nu_N(\XX),dr) + 1-\Delta_n\rb \end{equation*}
 where
 \begin{equation*} \tilde \wprho_{N,n}(D,A) \Def \frac{\int_{r\in [0,1]\cap A}\exp\left[\theta(1-r)\right]\wprho_{N,n}(D,dr)}{\int_{r\in [0,1]}\exp\left[\theta(1-r)\right]\wprho_{N,n}(D,dr)} \qquad A\in \Borel[0,1], \, D\in [0,1]\end{equation*}
 for all $N\in \N$ and $n\in \{1,2\dots N\}$.  In other words, the recovery rates for the names which have defaulted are independent with laws given by the $\tilde \wprho_{N,n}(\nu_N(\XX),\cdot)$'s.  In particular,
 \begin{equation*} \tilde \BE_N\left[L_N\big|\dilt\right] = \frac{1}{N}\sum_{n=1}^N \Delta_n\frac{\int_{r\in [0,1]}(1-r)\exp\left[\theta(1-r)\right]\wprho_{N,n}(\nu_N(\XX),dr)}{\int_{r\in [0,1]}\exp\left[\theta(1-r)\right]\wprho_{N,n}(\nu_N(\XX),dr)} = \Lambda'_{\nu_N}(\theta). \end{equation*}
Moreover,
 \begin{multline}\label{Eq:InitalLowerBound8} \tilde \BE_N\left[\left|L_N-\Lambda_{\nu_N}(\theta)\right|^2\big|\dilt\right] \\
 = \frac{1}{N^2}\sum_{n=1}^N \lb \int_{r\in [0,1]}(1-r)^2\tilde \wprho_{N,n}(\nu_N(\XX),dr)-\left(\int_{r\in [0,1]}(1-r)\tilde \wprho_{N,n}(\nu_N(\XX),dr)\right)^2\rb
 \le \frac{1}{N}. \end{multline}

 In a similar way, we next need to understand the statistics of the defaults under $\tilde \BP_N$.  For $\{\psi\}_{1\le n\le N}\subset \R$,
 \begin{equation*} \BE_N\left[\exp\left[\sqrt{-1}\sum_{n=1}^N \lb \psi_n\Delta_n + \phi(\ppp^{N,n})\Delta_n\rb \right]\right]
 = \prod_{n=1}^N \lb p^{N,n}\exp\left[\sqrt{-1}\psi_n + \phi(\ppp^{N,n})\right] + 1-p^{N,n}\rb. \end{equation*}
 Thus
 \begin{equation}\label{E:CDD}  \tilde \BE_N\left[\exp\left[\sqrt{-1}\sum_{n=1}^N \psi_n\Delta_n\right]\right]
 = \prod_{n=1}^N \lb \tilde p^{N,n}\exp\left[\sqrt{-1}\psi_n\right] + 1-\tilde p^{N,n}\rb \end{equation}
 where
 \begin{equation*} \tilde p^{N,n} = \frac{p^{N,n}e^{\phi(\ppp^{N,n})}}{p^{N,n}e^{\phi(\ppp^{N,n})}+1-p^{N,n}} = \lambda'_{p^{N,n}}(\phi(\ppp^{N,n})) = \frac{d\nu^*}{d\UU}(\ppp^{N,n}) \end{equation*}
 for all $N\in \N$ and $n\in \{1,2\dots N\}$.  In other words, under $\tilde \BP_N$ the defaults are independent with probabilities given by the $\tilde p^{N,n}$'s.

 \textbf{Step 3.} Let us go back to (\ref{Eq:InitalLowerBound6}). We want to show that $\varliminf_{N\to \infty}\tilde \BP_N(S_N)>0$, which will in turn follow if $\varlimsup_{N\to \infty}\tilde \BP_N(S^c_N)=0$.
To organize our thoughts, we write that
\begin{multline}\label{Eq:InitalLowerBound7} \tilde \BP_N(S_n^c) \le \tilde \BP_N\lb \nu_N\not\in \OO_3\rb + \tilde \BE_N\left[\tilde \BP_N\lb |L_N-\Lambda'_{\nu_N}(\theta)|\ge \eta_2\big|\dilt\rb \chi_{\{\nu_N\in \OO_3\}}\right] \\
 \le \tilde \BP_N\lb \nu_N\not\in \OO_3\rb + \frac{1}{\eta_2^2}\tilde \BE_N\left[\tilde \BE_N\left[\left|L_N-\Lambda'_{\nu_N}(\theta)\right|^2|\dilt\right] \chi_{\{\nu_N\in \OO_3\}}\right]. \end{multline}

Let us show that
 \begin{equation} \lim_{N\to \infty} \tilde \BE_N\left[\left|\int_{\ppp\in \XX}\Psi(\ppp)\nu_N(d\ppp)-\int_{\ppp\in \XX}\Psi(\ppp)\nu^*(d\ppp)\right|\right]=0. \label{Eq:InitalLowerBound9}\end{equation}

 Fix  $\Psi\in C(\XX)$.   We write
 \begin{equation*} \int_{\ppp\in \XX}\Psi(\ppp)\nu_N(d\ppp)-\int_{\ppp\in \XX}\Psi(\ppp)\nu^*(d\ppp) = \Err^N_1 + \Err^N_2 \end{equation*}
 where
 \begin{align*} \Err^N_1 &\Def \frac{1}{N}\sum_{n=1}^N\lb \Delta_n-\tilde p^{N,n}\rb \Psi(\ppp^{N,n}) \\
 \Err^N_2&\Def \frac{1}{N}\sum_{n=1}^N \frac{d\nu^*}{d\UU}(\ppp^{N,n})\Psi(\ppp^{N,n}) - \int_{\ppp\in \XX}\Psi(\ppp)\nu^*(d\ppp)
 = \int_{\ppp\in \XX}\frac{d\nu^*}{d\UU}(\ppp)\Psi(\ppp)\UU_N(d\ppp) - \int_{\ppp\in \XX}\frac{d\nu^*}{d\UU}(\ppp)\Psi(\ppp)\UU(d\ppp). \end{align*}
From \eqref{E:CDD} we have that $\tilde \BE_N[\Err^N_1]=0$; we also have by independence (Step 2) that
\begin{equation*} \tilde \BE_N\left[\left| \Err^N_1\right|^2\right] \le \frac{\sup_{\ppp\in \XX}|\Psi(\ppp)|^2}{N} \end{equation*}
The requirement that $\tfrac{d\nu^*}{d\UU}\in C(\XX)$ and Assumption \ref{A:LimitExists} ensure that $\lim_{N\to \infty}\Err^N_2=0$.

Combining things together, we have (\ref{Eq:InitalLowerBound9}).
Since $\Psi$ was an arbitrary element of $C(\XX)$ and $\XX$ is Polish, (\ref{Eq:InitalLowerBound9})  indeed implies (see \cite{MR1267569})
\begin{equation}
 \lim_{N\to \infty}\tilde \BP_N\{\nu_N\not \in \OO_3\}=0. \label{Eq:InitalLowerBound10}
\end{equation}
Using (\ref{Eq:InitalLowerBound8}) and (\ref{Eq:InitalLowerBound10}) we get by  (\ref{Eq:InitalLowerBound7}) that $\varlimsup_{N\to \infty}\tilde \BP_N(S^c_N)=0$. This and (\ref{Eq:InitalLowerBound6}) give us the statement of the proposition.
\end{proof}

We can now prove the full lower bound
\begin{proposition}\label{P:LowerBound} Let $G$ be an open subset of $[0,1]\times \SubProb$.  Then
\begin{equation*}\varliminf_{N\to \infty}\frac{1}{N}\ln \BP_N\lb Z_N\in G\rb \ge -\inf_{z\in G}I(z) \end{equation*}
\esquare
\end{proposition}
\begin{proof}  Fix $z=(\ell,\nu)\in G$.  If $I(z)<\infty$ and $\ell\in \ridom \Lambda^*_\nu$, then we get \eqref{E:lowerridom} from Proposition \ref{P:lowerridom}.
If $I(z)=\infty$, we of course again get \eqref{E:lowerridom}.  Finally, assume that $\ell\in \dom \Lambda^*_\nu \setminus \ridom \Lambda^*_\nu$.  We use
the fact that $\dom \Lambda^*_\nu \subset \overline{\ridom \Lambda^*_\nu}$ and convexity of $\ell \mapsto \Lambda^*_\nu(\ell)$.
Fix a relaxation parameter $\eta>0$.  Then there is an $\ell'\in \ridom \Lambda^*_\nu$ such that $(\ell',\nu)\in G$ and $\Lambda^*_\nu(\ell')<\Lambda^*_\nu(\ell)+\eta$ (see \cite{MR1619036}[Appendix A]).
Using Proposition \ref{P:lowerridom}, we get that
\begin{equation*} \varliminf_{N\to \infty}\frac{1}{N}\ln \BP_N\{ Z_N\in G\} \ge - I(\ell',\nu) \ge -I(z)-\eta. \end{equation*}
Letting $\eta\searrow 0$, we again get \eqref{E:lowerridom}.  Letting $z$ vary over $G$, we get the claim.
\end{proof}

\section{Large Deviations Upper Bound}\label{S:Upper}

The heart of the upper bound is an exponential Chebychev inequality.  We will mimic, as much as possible, the proof of the upper bound of Cram\'er's theorem.
The main result of this section is
\begin{proposition}\label{P:Upper} Fix any closed subset $F$ of $[0,1]\times \SubProb$.  Then
\begin{equation*} \varlimsup_{N\to \infty}\frac{1}{N}\ln \BP_N\lb Z_N\in F\rb \le -\inf_{z\in F}I(z). \end{equation*}
\esquare
\end{proposition}
\noindent Not surprisingly, we will first prove the bound for $F$ compact; we will then show enough exponential tightness to get to the full claim.

\begin{proposition}\label{P:UpperComp} Fix any compact subset $F$ of $[0,1]\times \SubProb$.  Then
\begin{equation*} \varlimsup_{N\to \infty}\frac{1}{N}\ln \BP_N\lb Z_N\in F\rb \le -\inf_{z\in F}I(z). \end{equation*}
\esquare
\end{proposition}
\begin{proof}  To begin, fix $s<\inf_{z\in F}I(z)$.  Fix also a relaxation parameter $\eta>0$.
For each $(\theta,\phi)\in \R\times C(\XX)$, define the set
\begin{equation*} \OO_{(\theta,\phi)} \Def \lb (\ell,\nu)\in [0,1]\times \SubProb:\, \theta \ell + \int_{\ppp\in \XX}\phi(\ppp)\nu(d\ppp) - \Lambda_\nu(\theta)> s+ \int_{\ppp=(p,\wprho)\in \XX}\lambda_p(\phi(\ppp))\UU(d\ppp) \rb \end{equation*}
(these open sets were used in the proof of Proposition \ref{P:LSC}).

Fix now a $z\in F$.  By definition of $I$ and Lemma \ref{L:Hrep}, we see that there is a $(\theta_z,\phi_z)\in \R\times C(\XX)$ such that
$z\in \OO_{(\theta_z,\phi_z)}$.  Since $(\ell,\nu)\mapsto \theta \ell + \int_{\ppp\in \XX}\phi_z(\ppp)\nu(d\ppp) - \Lambda_\nu(\theta)$ is continuous,
there is an open neighborhood $\OO^*_z$ of $z$ such that $\OO^*_z\subset \OO_{(\theta_z,\phi_z)}$ and such that
\begin{equation*} \theta_z \tilde \ell + \int_{\ppp\in \XX}\phi_z(\ppp)\tilde \nu(d\ppp) - \Lambda_{\tilde \nu}(\theta_z)> s+\int_{\ppp=(p,\wprho)\in \XX}\lambda_p(\phi_z(\ppp))\UU(d\ppp)
\end{equation*}
for all $(\tilde \ell,\tilde \nu)\in \OO^*_z$.  Thus
\begin{equation*} F \subset \bigcup_{z\in F}\OO^*_z, \end{equation*}
the compactness of $F$ implies that we can extract a finite subset $\ZZ$ of $F$ such that
\begin{equation*} F\subset \bigcup_{z\in \ZZ}\OO^*_z \end{equation*}
and thus
\begin{equation*} \BP_N\lb Z_N\in F\rb \le \sum_{z\in \ZZ}\BP_N\{Z_N\in \OO^*_z\}. \end{equation*}
Fix now $z\in \ZZ$.  We have that
\begin{multline*} \BP_N\{Z_N\in \OO^*_z\}
\le \BP_N\lb \theta_z L_N + \int_{\ppp\in \XX}\phi_z(\ppp)\nu_N(d\ppp) > s+ \Lambda_{\nu_N}(\theta_z) + \int_{\ppp=(p,\wprho)\in \XX}\lambda_p(\phi_z(\ppp))\UU(d\ppp)\rb \\
\le e^{-Ns} \BE_N\left[\exp\left[N\lb \theta_z L_N - \Lambda_{\nu_N}(\theta_z)\rb \right]\exp\left[N\lb \int_{\ppp\in \XX}\phi_z(\ppp)\nu_N(d\ppp) - \int_{\ppp=(p,\wprho)\in \XX}\lambda_p(\phi_z(\ppp))\UU(d\ppp)\rb\right]\right] \\
= e^{-Ns} \exp\left[N \lb \int_{\ppp=(p,\wprho)\in \XX}\lambda_p(\phi_z(\ppp))\UU_N(d\ppp)-\int_{\ppp=(p,\wprho)\in \XX}\lambda_p(\phi_z(\ppp))\UU(d\ppp)\rb\right] \end{multline*}
We have used here the fact that
\begin{equation*}\BE_N\left[\exp\left[N\left\{\theta_z L_N - \Lambda_{\nu_N}(\theta_z)\right\}\right]|\dilt\right]=1 \end{equation*}
(compare with (\ref{Eq:InitalLowerBound2a})) and that
\begin{equation*} \BE_N\left[\exp\left[N\int_{\ppp\in \XX}\phi_z(\ppp)\nu_N(d\ppp)\right]\right] = \exp\left[ N \int_{\ppp=(p,\wprho)\in \XX}\lambda_p(\phi_z(\ppp))\UU_N(d\ppp)\right]. \end{equation*}
Letting $N\to \infty$, we get that
\begin{equation*} \varlimsup_{N\to \infty}\frac{1}{N}\ln \BP_N\{Z_N\in \OO^*_z\} \le -s. \end{equation*}
This gives the claim.
\end{proof}

Let's next show that  $\nu_N$ is in a compact set.
\begin{proposition}[Exponential Tightness]\label{P:tightness}  For each $L>0$ there is a compact subset $\KComp_L$ of $\SubProb$ such that
\begin{equation*} \varlimsup_{N\to \infty}\frac{1}{N}\ln \BP_N\{\nu_n\not\in \KComp_L\} \le -L. \end{equation*}
\esquare
\end{proposition}
\begin{proof}  First note that Assumption \ref{A:LimitExists} implies that $\{\UU_N\}_{N\in \N}$ is tight.  Thus for each $j\in \N$, there is a compact subset $K_j$ of $\XX$
such that
\begin{equation*} \sup_{N\in \N}\UU_N(\XX\setminus K_j)\le \frac{1}{(L+j)^2}. \end{equation*}
We will define
\begin{equation*} \KComp_L \Def \overline{\lb \nu\in \SubProb:\, \text{$\nu(\XX\setminus K_j)\le \frac{1}{L+j}$ for all $j\in \N$}\rb}. \end{equation*}
Then $\KComp_L$ is compact, and we have that
\begin{multline*} \BP_N\{\nu_N\not \in \KComp_L\} \le \sum_{j=1}^\infty \BP_N\lb \nu_N(\XX\setminus K_j)\ge \frac{1}{L+j}\rb \\
\le \sum_{j=1}^\infty \BP_N\lb N(L+j)^2\nu_N(\XX\setminus K_j)\ge N(L+j)\rb\\
\le \sum_{j=1}^\infty \exp\left[-N(L+j)\right]\BE_N\left[\exp\left[N(L+j)^2\nu_N(\XX\setminus K_j)\right]\right]\end{multline*}
We now compute that
\begin{multline*} \BE_N\left[\exp\left[N(L+j)^2\nu_N(\XX\setminus K_j)\right]\right]
= \prod_{n=1}^N \BE_N\left[\exp\left[N(L+j)^2\sum_{n=1}^N \Delta_n \chi_{\XX\setminus K_j}(\ppp^{N,n})\right]\right] \\
= \exp\left[N\sum_{n=1}^N \lambda_{p^{N,n}}\left( (L+j)^2\chi_{\XX\setminus K_j}(\ppp^{N,n})\right)\right]
\le \exp\left[ N(L+j)^2 \UU_N(\XX\setminus K_j)\right]\le e^N. \end{multline*}
We have used here the calculation that for $\theta>0$,
\begin{equation*} \lambda_p(\theta)\le \ln \left(pe^\theta + (1-p)e^\theta\right) = \theta. \end{equation*}
Combining things together, we get that
\begin{equation*} \BP_N\{\nu_N\not\in \KComp_L\} \le \sum_{j=1}^\infty e^{-N(L+j)}e^N = e^{-NL}\sum_{j=1}^\infty e^{-N(j-1)} \le e^{-NL}\sum_{j=1}^\infty e^{-(j-1)} =  \frac{e^{-NL}}{1-e^{-1}}. \end{equation*}

\end{proof}

We can now get the full upper bound.
\begin{proof}[Proof of Proposition \ref{P:Upper}]  Fix $L>0$.  Then
\begin{equation*} \BP_N\{Z_N\in F\} \le \BP_N\{Z_N\in F,\, \nu_N\in \KComp_L\} + \BP_N\{\nu_N\not\in \KComp_L\}. \end{equation*}
We use Proposition \ref{P:UpperComp} on the first term together with the fact that $[0,1]\times \KComp_L$ is compact.
We use Proposition \ref{P:tightness} on the second term.  Noting that (Lemma 1.2.15 in \cite{MR1619036})
\begin{equation*}
\varlimsup_{N\to \infty}\frac{1}{N}\ln\left(A_N+B_N\right)=\max\{\varlimsup_{N\to \infty}\frac{1}{N}\ln A_N, \varlimsup_{N\to \infty}\frac{1}{N}\ln B_N\}
\end{equation*}
we get
\begin{equation*} \varlimsup_{N\to \infty}\frac{1}{N}\ln \BP_N\{Z_N\in F,\, \nu_N\in \KComp_L\} \le -\inf_{\substack{z=(\ell,\nu)\in F\\ \nu\in \KComp_L}}I(z)\le -\inf_{z\in F}I(z). \end{equation*}
\end{proof}

\section{Alternative Expression for the Rate Function}\label{S:AlternateExpression}

In this section, we discuss the alternative expression for the rate function $I'$ of Theorem \ref{T:Main} given by Theorem \ref{T:AlternativeExpression}.  In particular, this alternative representation shows that $I'(\ell)$ has a natural interpretation as the favored way to rearrange recoveries and losses among the different types.  In addition to providing intuitive insight, this alternative expression suggests numerical schemes for computing the rate function.  We will rigorously verify that the alternative expression is correct, but will be heuristic in our discussion of the numerical schemes.

We defer the proof of  Theorem \ref{T:AlternativeExpression} to the end of this section and we first study the variational problem \eqref{E:IDef} using a Lagrange multiplier approach. Even though an explicit expression is usually not
available, one can use numerical optimization
techniques to calculate the quantities involved. In order to do
that, we firstly  recall that we can rewrite $J'$ of \eqref{E:IDef}
as a two-stage minimization problem, see expression \eqref{E:IIDef}.

This naturally suggests an analysis via a Lagrangian.  Define
\begin{multline*} L(\Phi,\Psi,\lambda_1,\lambda_2) = \int_{\ppp=(p,\wprho)\in \XX} \lb \Phi(p)I_\wprho\left(\Psi(\ppp),D\right) + \hbar_p(\Phi(\ppp))\rb \UU(d\ppp) \\
-\lambda_1\lb \int_{\ppp\in \XX} \Phi(\ppp)\Psi(\ppp)\UU(d\ppp) -\ell\rb
-\lambda_2\lb \int_{\ppp\in \XX} \Phi(\ppp)\UU(d\ppp) -D\rb. \end{multline*}
Let's assume that $\Phi^*$ and $\Psi^*$ are the minimizers.  Let's also assume that $I_\wprho(\cdot,D)$ is differentiable for all $\ppp=(p,\wprho)$ in the support of $\UU$.  We
should then have that for every $\eta_1$ and $\eta_2$ in $\BB=B(\XX;[0,1])$,
\begin{align*}
\int_{\ppp=(p,\wprho)\in \XX} \eta_1(\ppp)\lb I_\wprho\left(\Psi^*(\ppp),D\right) + \hbar'_p(\Phi^*(\ppp))-\lambda_1 \Psi^*(\ppp)- \lambda_2 \rb \UU(d\ppp)&=0 \\
\int_{\ppp=(p,\wprho)\in \XX} \eta_2(\ppp)\Phi^*(\ppp)\lb I'_\wprho\left(\Psi^*(\ppp),D\right) - \lambda_1\rb \UU(d\ppp)&=0 \end{align*}
Ignoring any complications which would arise
on the set where $\Phi^*=0$, we should then have that
\begin{align*} I_\wprho\left(\Psi^*(\ppp),D\right) + \hbar'_p(\Phi^*(\ppp))&=\lambda_1\Psi^*(\ppp)+\lambda_2 \\
I'_\wprho\left(\Psi^*(\ppp),D\right) &=\lambda_1 \end{align*}
for all $\ppp=(p,\wprho)\in \XX$.
This is a triangular system; the first equation depends on both
$\lambda_1$ and $\lambda_2$, but the second depends only on $\lambda_1$.
Recalling now \eqref{E:LFT} and the structure of Legendre-Fenchel transforms,
we should have that
\begin{gather*} M'_{\wprho}(\lambda_1,D)=\Psi^*(\ppp)\\
\hbar'_p(\Phi^*(\ppp)) = \lambda_2+\lambda_1M'_\wprho(\lambda_1,D)-I_\wprho(\Psi^*(\ppp),D)=\lambda_2+M_\wprho(\lambda_1,D) \end{gather*}
for all $\ppp=(p,\wprho)\in \XX$.
We can then invert this.  This leads us to the following.  Define
\begin{align*} \Phi_{\lambda_1,\lambda_2,D}(p,\wprho) &\Def \frac{p e^{\lambda_2+M_{\wprho}(\lambda_1,D)}}{1-p+p e^{\lambda_2+M_{\wprho}(\lambda_1,D)}} \qquad \lambda_1,\lambda_2\in \R,\, (p,\wprho)\in \XX\\
\Psi_{\lambda_1,D}(\wprho)&\Def M'_{\wprho}(\lambda_1,D) \qquad \lambda_1\in \R,\, \wprho\in C([0,1];\Pint)
\end{align*}
where $(\lambda_1,\lambda_2)=(\lambda_1(\ell,D,\UU),\lambda_2(\ell,D,\UU))$ is such that
\begin{align*} \int_{\ppp\in\XX}\Phi_{\lambda_1,\lambda_2,D}(\ppp)\UU(d\ppp) &=D\\
\int_{\ppp\in\XX}\Phi_{\lambda_1,\lambda_2,D}(\ppp)\Psi_{\lambda_1,\lambda_2,D}(\ppp)\UU(d\ppp) &=\ell. \end{align*}

We conclude this section with the rigorous proof of the alternate representation.
\begin{proof}[Proof of Theorem \ref{T:AlternativeExpression}]
First, we prove that $J'(\ell)\geq I'(\ell)$. Consider any $\Phi$ and $\Psi \in \BB$ such that
\begin{equation}\label{E:require}\int_{\ppp\in\XX}\Phi(\ppp)\Psi(\ppp)\UU(d\ppp)=\ell.\end{equation}
For any $\theta\in \R$,
\begin{align*}
&\int_{\ppp=(p,\wprho)\in\XX}\lb\Phi(\ppp)
I_\wprho\left(\Psi(\ppp),\int_{\ppp=(p,\wprho)\in \XX}\Phi(\ppp)\UU(d\ppp)\right)+
\hbar_p(\Phi(\ppp)) \rb
\UU(d\ppp)\\
&\qquad =\int_{\ppp=(p,\wprho)\in\XX}\lb\Phi(\ppp) \sup_{\theta'\in\R}\lb
\theta'\Psi(\ppp)-M_{\wprho}\left(\theta', \int_{\ppp=(p,\wprho)\in
\XX}\Phi(\ppp)\UU(d\ppp) \right)\rb + \hbar_p(\Phi(\ppp)) \rb
\UU(d\ppp)\\
&\qquad \geq\int_{\ppp=(p,\wprho)\in\XX}\lb\Phi(\ppp) \lb
\theta\Psi(\ppp)-M_{\wprho}\left(\theta, \int_{\ppp=(p,\wprho)\in
\XX}\Phi(\ppp)\UU(d\ppp) \right)\rb + \hbar_p(\Phi(\ppp)) \rb
\UU(d\ppp)\\
&\qquad =\theta\ell-\int_{\ppp=(p,\wprho)\in\XX}\lb\Phi(\ppp) M_{\wprho}\left(\theta,
\int_{\ppp\in \XX}\Phi(\ppp)\UU(d\ppp) \right) + \hbar_p(\Phi(\ppp)) \rb
\UU(d\ppp).
\end{align*}
Define $\nu\in \SubProb$ as
\begin{equation*} \nu(A)\Def \int_{\ppp\in A}\Phi(\ppp)\UU(d\ppp);\qquad A\in \Borel(\XX) \end{equation*}
then
\begin{multline*}
\int_{\ppp=(p,\wprho)\in\XX}\lb\Phi(\ppp) I_\wprho\left(\Psi(\ppp),\int_{\ppp\in
\XX}\Phi(\ppp)\UU(d\ppp)\right)+ \hbar_p(\Phi(\ppp)) \rb
\UU(d\ppp)\\
\geq\theta\ell-\int_{\ppp=(p,\wprho)\in\XX}
M_{\wprho}(\theta,\nu(\XX))\nu(d\ppp) + \int_{\ppp=(p,\wprho)\in\XX}
\hbar_p\left(\frac{d\nu}{d\UU}(\ppp)\right) \UU(d\ppp).
\end{multline*}
Varying $\theta$, we get that
\begin{equation*}
\int_{\ppp=(p,\wprho)\in\XX}\lb\Phi(\ppp) I_\wprho\left(\Psi(\ppp),\int_{\ppp=(p,\wprho)\in
\XX}\Phi(\ppp)\UU(d\ppp)\right)+ \hbar_p(\Phi(\ppp)) \rb
\UU(d\ppp)\geq \Lambda_\nu^*(\ell)+H(\nu)\geq I'(\ell)
\end{equation*}
and then varying $\Phi$ and $\Psi$  in $\BB$ (such that \eqref{E:require} holds), we get that $J'(\ell)\geq I'(\ell)$.

To show that $I'(\ell)\geq J'(\ell)$, fix $\nu\in \SubProb$ such  that $\nu\ll\UU$.   We want to show that
\begin{equation}\label{E:Lcomp} \Lambda_\nu^*(\ell)+H(\nu)\ge J'(\ell). \end{equation}

For all $\wprho\in C([0,1];\Pint)$, define
\begin{align*}
\alpha_-(\wprho,D)&\Def\inf\{1-r\in[0,1]: r\in\supp \wprho(D,\cdot)\}\\
\alpha_+(\wprho,D)&\Def\sup\{1-r\in[0,1]: r\in\supp \wprho(D,\cdot)\}.\
\end{align*}
Dominated convergence implies that
\begin{align*}
\lim_{\theta\to-\infty}\Lambda'_\nu(\theta)&=\bar{\alpha}_-\Def \int_{\ppp=(p,\wprho)\in\XX} \alpha_-(\wprho,\nu(\XX))\nu(d\ppp)\\
\lim_{\theta\to\infty}\Lambda'_\nu(\theta)&=\bar{\alpha}_+\Def \int_{\ppp=(p,\wprho)\in\XX}\alpha_+(\wprho,\nu(\XX))\nu(d\ppp).
\end{align*}
From \eqref{E:LambdaI} and the monotonicity of moment generating functions,
we can see that $\Lambda_\nu$ is nondecreasing; thus $(\bar{\alpha}_-, \bar{\alpha}_+)\in\Lambda'_\nu(\R)$.
This leads to three possible cases.

\textbf{Case 1}:
Assume that $\ell\in (\bar{\alpha}_-, \bar{\alpha}_+)$, and let $\theta^*\in\R$ be such that $\Lambda'_\nu(\theta^*)=\ell$; i.e.,
\begin{equation}\label{E:lambdader}\int_{\ppp=(p,\wprho)\in \XX} M'_{\wprho}(\theta^*,\nu(\XX)) \nu(d\ppp)=\ell\end{equation}
Then
\begin{align*}
\Lambda^*_\nu(\ell)+H(\nu)&=
\sup_{\theta\in\R}\lb\theta\ell -
\int_{\ppp=(p,\wprho)\in\XX}M_{\wprho}(\theta,\nu(\XX))\nu(d\ppp)\rb+\int_{\ppp=(p,\wprho)\in\UU}\hbar_p\left(\frac{d\nu}{d\UU}(\ppp)\right)\UU(d\ppp)\\
&\geq\theta^*\ell -
\int_{\ppp=(p,\wprho)\in\XX}M_{\wprho}(\theta^*,\nu(\XX))\nu(d\ppp)+\int_{\ppp=(p,\wprho)\in\UU}\hbar_p\left(\frac{d\nu}{d\UU}(\ppp)\right)\UU(d\ppp)\\
&= \theta^*\Lambda'_\nu(\theta^*) -
\int_{\ppp=(p,\wprho)\in\XX}M_{\wprho}(\theta^*,\nu(\XX))\nu(d\ppp)+\int_{\ppp=(p,\wprho)\in\UU}\hbar_p\left(\frac{d\nu}{d\UU}(\ppp)\right)\UU(d\ppp)\\
&=
\int_{\ppp=(p,\wprho)\in\XX}\left(\theta^* M'_{\wprho}(\theta^*,\nu(\XX)) -M_{\wprho}(\theta^*,\nu(\XX))\right)\nu(d\ppp)+\int_{\ppp=(p,\wprho)\in\UU}\hbar_p\left(\frac{d\nu}{d\UU}(\ppp)\right)\UU(d\ppp)\\
&=
\int_{\ppp=(p,\wprho)\in\XX} I_{\wprho}(M'_{\wprho}(\theta^*,\nu(\XX)),\nu(\XX)) \nu(d\ppp)+\int_{\ppp=(p,\wprho)\in\UU}\hbar_p\left(\frac{d\nu}{d\UU}(\ppp)\right)\UU(d\ppp).
\end{align*}
Define now $\Phi(\ppp)\Def \frac{d\nu}{d\UU}(\ppp)$ and
$\Psi(\ppp)\Def M'_{\wprho}(\theta^*,\nu(\XX))$.  By (\ref{E:smalla}) and (\ref{Eq:Duals_M_I}) respectively we have that $\Phi,\Psi\in\BB$. Then \eqref{E:lambdader} is exactly that $\int_{\ppp\in\XX}\Phi(\ppp)\Psi(\ppp)\UU(d\ppp)=\ell$.  Thus
\begin{equation*}
\Lambda^*_\nu(\ell)+H(\nu)\geq
\int_{\ppp=(p,\wprho)\in\XX}\left[\Phi(\ppp)I_{\wprho}\left(\Psi(\ppp),
\int_{\ppp\in \XX}\Phi(\ppp)\UU(d\ppp)
\right)+\hbar_p\left(\Phi(\ppp)\right)\right]\UU(d\ppp)\geq J'(\ell).\end{equation*}
This is exactly (\ref{E:Lcomp}).

\textbf{Case 2}:
Assume next that $\ell\in [\bar\alpha_+, 1]$.  For every $\wprho\in C([0,1];\Pint)$, define
\begin{equation*} \Err^\wprho_+(\theta) \Def M_\wprho(\theta,\nu(\XX))-\theta \alpha_+(\wprho,\nu(\XX))
=\ln \int_{r\in[0,1]}e^{-\theta \left(\alpha_+(\wprho,\nu(\XX))-(1-r)\right)} \wprho(\nu(\XX),dr).
\end{equation*}
for all $\theta\in\R$; thus
\begin{align*} M_\wprho(\theta,\nu(\XX)) &= \theta \alpha_+(\wprho,\nu(\XX)) + \Err^\wprho_+(\theta) \qquad \wprho\in C([0,1];\Pint) \\
\Lambda_\nu(\theta) &= \theta \bar \alpha_+ + \int_{\ppp=(p,\wprho)\in \XX}\Err^\wprho_+(\theta)\nu(d\ppp) \end{align*}
for all $\theta\in \R$.  For all $\wprho\in C([0,1];\Pint)$ and  $(1-r)\in\supp \wprho(\nu(\XX),\cdot)$,  the mapping $\theta\mapsto e^{-\theta \left(\alpha_+(\wprho,\nu(\XX))-(1-r)\right)}$ is decreasing and maps $[0,\infty)$ into $(0,1]$.  Monotone convergence implies that
\begin{equation*} \lim_{\theta\to \infty}\Err^\wprho_+(\theta) = \ln \wprho\{\nu(\XX),1-\alpha_+(\wprho,\nu(\XX))\}. \end{equation*}
If $\ell>\bar\alpha_+$, then we can use the fact that $\int_{\ppp=(p,\wprho)\in \XX}\Err^\wprho_+(\theta)\nu(d\ppp)\le 0$ for all $\theta>0$ to see that
\begin{equation*}
\Lambda^*_\nu(\ell)\geq \varlimsup_{\theta\to \infty}\lb \theta (\ell-\bar\alpha_+)-\int_{\ppp=(p,\wprho)\in \XX}\Err^\wprho_+(\theta)\nu(d\ppp)\rb\geq \varlimsup_{\theta\to \infty}\lb \theta (\ell-\bar\alpha_+)\rb=\infty\geq J'(\ell).
\end{equation*}
If $\ell=\bar\alpha_+$, then by the monotonicity of the $\Err^\wprho_+$'s,
\begin{equation}I_\wprho(\alpha_+(\wprho,\nu(\XX)),\nu(\XX)) = \sup_{\theta\in\R}\lb -\Err^\wprho_+(\theta)\rb= \varlimsup_{\theta\to \infty}\lb -\Err^\wprho_+(\theta)\rb = \ln \left(\frac{1}{\wprho\{\nu(\XX),1-\alpha_+(\wprho,\nu(\XX))\}}\right) \end{equation}
for all $\wprho\in C([0,1];\Pint)$, and
\begin{multline*} \Lambda^*_\nu(\ell)= \sup_{\theta\in\R}\lb -\int_{\ppp=(p,\wprho)\in \XX}\Err^\wprho_+(\theta)\nu(d\ppp)\rb= \varlimsup_{\theta\to \infty}\lb -\int_{\ppp=(p,\wprho)\in \XX}\Err^\wprho_+(\theta)\nu(d\ppp)\rb\\
=\int_{\ppp=(p,\wprho)\in\XX}\ln \left(\frac{1}{\wprho\{\nu(\XX),1-\alpha_+(\wprho,\nu(\XX))\}}\right) \nu(d\ppp) =\int_{\ppp=(p,\wprho)\in\XX}I_\wprho(\alpha_+(\wprho,\nu(\XX)))\nu(d\ppp).
 \end{multline*}
Defining $\Phi(\ppp)\Def \frac{d\nu}{d\UU}(\ppp)$ and $\Psi(\ppp)=\alpha_+(\wprho,\nu(\XX))$, we have that
\begin{equation*} \int_{\ppp\in \XX}\Phi(\ppp)\Psi(\ppp)\UU(d\ppp) = \int_{\ppp=(p,\wprho)\in\XX}\alpha_+(\wprho,\nu(\XX))\nu(d\ppp) =\ell. \end{equation*}
Collecting things together, we see that if $\ell=\bar\alpha_+$, we again get \eqref{E:Lcomp}.

\textbf{Case $3$}:
We finally assume that $\ell\in [0,\bar\alpha_-]$.  The calculations are very similar to those of Case 2.
For every $\wprho\in C([0,1];\Pint)$, define
\begin{equation*} \Err^\wprho_-(\theta) \Def M_\wprho(\theta,\nu(\XX))-\theta \alpha_-(\wprho,\nu(\XX))
=\ln \int_{r\in[0,1]}e^{\theta \left((1-r)-\alpha_-(\wprho,\nu(\XX))\right)} \wprho(\nu(\XX),dr).
\end{equation*}
for all $\theta\in\R$, so that
\begin{align*} M_\wprho(\theta,\nu(\XX)) &= \theta \alpha_-(\wprho,\nu(\XX)) + \Err^\wprho_-(\theta) \wprho\in C([0,1];\Pint) \\
\Lambda_\nu(\theta) &= \theta \bar \alpha_- + \int_{\ppp=(p,\wprho)\in \XX}\Err^\wprho_-(\theta)\nu(d\ppp)  \end{align*}
for all $\theta\in \R$.  For all $\wprho\in C([0,1];\Pint)$ and  $(1-r)\in\supp \wprho(\nu(\XX),\cdot)$,  the mapping $\theta\mapsto e^{\theta \left((1-r)-\alpha_-(\wprho,\nu(\XX))\right)}$ is increasing and maps $(-\infty,0]$ into $(0,1]$.  Monotone convergence implies that
\begin{equation*} \lim_{\theta\to -\infty}\Err^\wprho_-(\theta) = \ln \wprho\{\nu(\XX),1-\alpha_-(\wprho,\nu(\XX))\}. \end{equation*}
If $\ell<\bar\alpha_-$, then we can use the fact that $\int_{\ppp=(p,\wprho)\in \XX}\Err^\wprho_-(\theta)\nu(d\ppp)\ge 0$ for all $\theta<0$ to see that
\begin{equation*}
\Lambda^*_\nu(\ell)\geq \varlimsup_{\theta\to -\infty}\lb \theta (\ell-\bar\alpha_-)-\int_{\ppp=(p,\wprho)\in \XX}\Err^\wprho_-(\theta)\nu(d\ppp)\rb\geq \varlimsup_{\theta\to -\infty}\lb \theta (\ell-\bar\alpha_-)\rb=\infty\geq J'(\ell).
\end{equation*}
If $\ell=\bar\alpha_-$, then by the monotonicity of the $\Err^\wprho_-$'s,
\begin{equation*}I_\wprho(\alpha_-(\wprho,\nu(\XX)),\nu(\XX)) = \sup_{\theta\in\R}\lb -\Err^\wprho_-(\theta)\rb= \lim_{\theta\to -\infty}\lb -\Err^\wprho_-(\theta)\rb = \ln \left(\frac{1}{\wprho\{1-\alpha_-(\wprho,\nu(\XX))\}}\right)\end{equation*}
for all $\wprho\in C([0,1];\Pint)$, and
\begin{multline*} \Lambda^*_\nu(\ell)= \sup_{\theta\in\R}\lb -\int_{\ppp=(p,\wprho)\in \XX}\Err^\wprho_-(\theta)\nu(d\ppp)\rb= \lim_{\theta\to -\infty}\lb -\int_{\ppp=(p,\wprho)\in \XX}\Err^\wprho_-(\theta)\nu(d\ppp)\rb\\
=\int_{\ppp=(p,\wprho)\in\XX}\ln \left(\frac{1}{\wprho\{\nu(\XX),1-\alpha_-(\wprho,\nu(\XX))\}}\right) \nu(d\ppp)=\int_{\ppp=(p,\wprho)\in\XX}I_\wprho(\alpha_-(\wprho,\nu(\XX)))\nu(d\ppp).
\end{multline*}
Defining $\Phi(\ppp)\Def \frac{d\nu}{d\UU}(\ppp)$ and $\Psi(\ppp)=\alpha_-(\wprho,\nu(\XX))$, we have that
\begin{equation*} \int_{\ppp\in \XX}\Phi(\ppp)\Psi(\ppp)\UU(d\ppp) = \int_{\ppp=(p,\wprho)\in\XX}\alpha_-(\wprho,\nu(\XX))\nu(d\ppp) =\ell \end{equation*}
again implying \eqref{E:Lcomp}.

Collecting things together, we have \eqref{E:IDef}.  We get \eqref{E:IIDef} by definng $D\Def\int_{\ppp\in \XX}\Phi(\ppp)\UU(d\ppp)$.  Note that since $\Phi$ and $\Psi$
both take values in $[0,1]$,
\begin{equation*} \int_{\ppp\in \XX}\Psi(\ppp)\Phi(\ppp)\UU(d\ppp)\le \int_{\ppp\in \XX}\Phi(\ppp)\UU(d\ppp). \end{equation*}
This allows us to restrict the minimization in $D$ to the interval $[\ell,1]$.
\end{proof}

\section{Detailed structure of $H$}\label{S:AuxiliaryLemmas}

In this section we  prove Lemmas \ref{L:approximation} and \ref{L:Hrep}. The discussion in this section is somewhat technical.

\noindent
Recall the quantities defined in Definition \ref{E:LFT}.   Note that
\begin{equation}\label{E:hbarcalcs}\hbar_0(x) = \begin{cases} 0 &\text{if $x=0$} \\
\infty &\text{else}\end{cases} \qquad \text{and}\qquad \hbar_1(x) = \begin{cases} 0 &\text{if $x=1$} \\
\infty &\text{else}\end{cases} \end{equation}
Thus if $H(\nu)<\infty$, we can restrict the region of integration to get that
\begin{equation*}\int_{\ppp=(p,\wprho)\in \XX}\chi_{\{0,1\}}(p)\hbar_p\left(\frac{d\nu}{d\UU}(\ppp)\right)\UU(d\ppp)<\infty \end{equation*}
so in fact
\begin{equation}\label{E:smallb}\begin{aligned}\UU\lb \ppp=(p,\wprho)\in \XX:\, \text{$p=0$ and $\frac{d\nu}{d\UU}(\ppp)\not = 0$}\rb &=0 \\
\UU\lb \ppp=(p,\wprho)\in \XX:\, \text{$p=1$ and $\frac{d\nu}{d\UU}(\ppp)\not = 1$}\rb &=0.\end{aligned}\end{equation}

Fix $\nu\in \SubProb$ such that $H(\nu)<\infty$.  The main technical challenges in both proofs is to stay away from the singularities in $\hbar_p$ and $\hbar'_p$.  Note that
\begin{equation*}\hbar'_p(x) = \ln \left(\frac{x}{1-x}\frac{1-p}{p}\right), \qquad x,p\in (0,1)\end{equation*}
and keeping \eqref{E:hbarcalcs} in mind, we thus need to be careful near $p\in \{0,1\}$, and for $(x,p)\in \{0,1\}\times (0,1)$.

To start, let's note some implications of the assumption that $H(\nu)<\infty$.  Clearly $\nu\ll \UU$.  Secondly,
\begin{equation}\label{E:smalla} \UU\lb \ppp\in \XX:\, \frac{d\nu}{d\UU}(\ppp)>1\rb =0. \end{equation}
Let's now do the following.  Fix $N\in \N$.  Define
\begin{equation}\label{E:XiDef} \xi_N(\ppp) \Def \begin{cases} p &\text{if $p\not\in \left[\frac{1}{N},1-\frac{1}{N}\right]$} \\
\frac{d\nu}{d\UU}(\ppp) &\text{if $p\in \left[\frac{1}{N},1-\frac{1}{N}\right]$ and $\frac{d\nu}{d\UU}(\ppp)\in \left(\frac{1}{N},1-\frac{1}{N}\right)$} \\
\frac{1}{N}& \text{if $p\in \left[\frac{1}{N},1-\frac{1}{N}\right]$ and $\frac{d\nu}{d\UU}(\ppp)\le \frac{1}{N}$} \\
1-\frac{1}{N}& \text{if $p\in \left[\frac{1}{N},1-\frac{1}{N}\right]$ and $\frac{d\nu}{d\UU}(\ppp)\ge 1-\frac{1}{N}$} \end{cases}\end{equation}
Clearly $0\le \xi_N\le 1$, so we can define $\nu_N\in \SubProb$ as
\begin{equation*}\nu_N(A) \Def \int_{\ppp\in A}\xi_N(\ppp)\UU(d\ppp). \qquad A\in \Borel(\XX) \end{equation*}
In light of \eqref{E:smalla} and \eqref{E:smallb}, $\lim_{N\to \infty}\xi_N=\frac{d\nu}{d\UU}$ $\UU$-a.s., so it follows that $\lim_{N\to \infty}\nu_N=\nu$.
We next compute that
\begin{equation*} \hbar_p(\xi_N(\ppp)) = \begin{cases} 0 &\text{if $p\not\in \left[\frac1{N},1-\frac{1}{N}\right]$} \\
\hbar_p\left(\frac{d\nu}{d\UU}(\ppp)\right) &\text{if $\frac1{N}\le p\le 1-\frac{1}{N}$ and $\frac{d\nu}{d\UU}(\ppp)\in \left(\frac{1}{N},1-\frac{1}{N}\right)$} \\
\hbar_p\left(\frac{1}{N}\right)& \text{if $\frac1{N}\le p\le 1-\frac{1}{N}$ and $\frac{d\nu}{d\UU}(\ppp)\le \frac{1}{N}$} \\
\hbar_p\left(1-\frac{1}{N}\right)& \text{if $\frac1{N}\le p\le 1-\frac{1}{N}$ and $\frac{d\nu}{d\UU}(\ppp)\ge 1-\frac{1}{N}$} \end{cases}\end{equation*}
Using again \eqref{E:smalla} and \eqref{E:smallb}, we have that $\lim_{N\to \infty}\hbar_p(\xi_N(\ppp))=\hbar_p\left(\frac{d\nu}{d\UU}(\ppp)\right)$
for $\UU$-almost-all $\ppp=(p,\wprho)\in \XX$.  If $p\in \left[\frac1{N},1-\frac{1}{N}\right]$, then $\hbar_p$ is increasing on $[p,1]\supset \left[1-\frac1{N},1\right]$
and decreasing on $[0,p]\supset\left[0,\frac{1}{N}\right]$.  Thus $\hbar_p(\xi_N(\ppp))\le \hbar_p\left(\frac{d\nu}{d\UU}(\ppp)\right)$
for $\UU$-almost-all $\ppp=(p,\wprho)\in \XX$.  Dominated convergence thus implies that $\lim_{N\to \infty}H(\nu_N)=H(\nu)$.

\begin{proof}[Proof of Lemma \ref{L:approximation}]
Fix $N\in N$; we want to approximate $\xi_N$ by ``nice'' elements of $C(\XX)$.  Note that
\begin{equation*} \xi_N(\ppp) = p\chi_{[0,1]\setminus [N^{-1},1-N^{-1}]}(p) + \chi_{[N^{-1},1-N^{-1}]}(p)\xi_N(\ppp). \end{equation*}
Since $\UU$ is regular (recall that $\XX$ is Polish), we can approximate $\ppp=(p,\wprho)\mapsto \chi_{[N^{-1},1-N^{-1}]}(p)\xi_N(\ppp)$
by elements of $C(\XX)$.   From \eqref{E:XiDef}, we have that $N^{-1}\le \xi_N(\ppp)\le 1-N^{-1}$ if $\ppp=(p,\wprho)\in \XX$ is such that $N^{-1}\le p\le 1-N^{-1}$,
so we can truncate these approximations at $N^{-1}$ and $1-N^{-1}$ without any loss.  Namely,
there is a sequence $(\xi^{1,\eps})_{\eps>0}$ in $C(\XX)$ such that
\begin{equation} \label{E:xitilde} \begin{gathered} N^{-1}\le \xi^{1,\eps}\le 1-N^{-1} \\
\lim_{\eps\searrow 0}\int_{\ppp=(p,\wprho)\in \XX}\left|\xi^{1,\eps}(\ppp)-\chi_{[N^{-1},1-N^{-1}]}(p)\xi_N(\ppp)\right|\UU(d\ppp)=0. \end{gathered}\end{equation}
For each $\eps>0$, let $\varphi_\eps\in C([0,1];[0,1])$ be such that $\varphi_\eps(u)=1$ if $u\in [N^{-1},1-N^{-1}]$ and $\varphi_\eps(u)=0$ if $u\in [0,1]\setminus [N^{-1}-\eps,1-N^{-1}+\eps]$.  For each $\eps>0$, define
\begin{equation*} \xi^{2,\eps}(\ppp) \Def p\lb 1- \varphi_\eps(p)\rb + \xi^{1,\eps}(\ppp)\varphi_\eps(p) \end{equation*}
for all $\ppp=(p,\wprho)\in \XX$.  Then $\xi^{2,\eps}\in C(\XX)$ for all $\eps>0$.  We also have that
\begin{multline*} \int_{\ppp=(p,\wprho)\in \XX}\left|\xi^{2,\eps}(\ppp)-\xi_N(\ppp)\right|\UU(d\ppp)\\
\le \UU\lb \ppp=(p,\wprho)\in \XX:  p\in [N^{-1}+\eps,1-N^{-1}+\eps]\setminus [N^{-1},1-N^{-1}]\rb \\
+ \int_{\ppp=(p,\wprho)\in \XX} \chi_{[N^{-1},1-N^{-1}]}(p)\left|\xi^{1,\eps}(\ppp)-\xi_N(\ppp)\right|\UU(d\ppp). \end{multline*}
Dominated convergence and \eqref{E:xitilde} then ensure that
\begin{equation}\label{E:Goofy} \lim_{\eps \to 0}\int_{\ppp=(p,\wprho)\in \XX}\left|\xi^{2,\eps}(\ppp)-\xi_N(\ppp)\right| \UU(d\ppp)=0. \end{equation}
Clearly $0\le \xi^{2,\eps}\le 1$, so we can define $\nu_{N,\eps}\in \SubProb$ as
\begin{equation*} \nu_{N,\eps}(A)\Def \int_{\ppp\in A}\xi^{2,\eps}(\ppp)\UU(d\ppp). \qquad A\in \Borel(\XX) \end{equation*}
Thanks to \eqref{E:Goofy}, we have that $\lim_{\eps \to 0}\nu_{N,\eps}=\nu_N$.  Note next that for $\ppp=(p,\wprho)\in \XX$ such that $p\in [0,1]\setminus [N^{-1}-\eps,1-N^{-1}+\eps]$,
\begin{equation*} \hbar_p(\xi^{2,\eps}(\ppp)) - \hbar_p(\xi_N(\ppp)) = \hbar_p(p)-\hbar_p(p)=0.\end{equation*}
If $\ppp=(p,\wprho)\in \XX$ is such that $p\in [N^{-1}-\eps,1-N^{-1}+\eps]$, then
\begin{equation} \label{E:stuff} N^{-1}-\eps\le \xi^{2,\eps}(\ppp)\le 1-N^{-1}+\eps, \end{equation}
so if $\eps<1/(2N)$,
\begin{equation*} \left|\hbar_p(\xi^{2,\eps}(\ppp))-\hbar_p(\xi_N(\ppp))\right| \le \vkap\left|\xi^{2,\eps}(\ppp)-\xi_N(\ppp)\right| \end{equation*}
where
\begin{equation*} \vkap\Def \sup\lb |\hbar'_p(x)|:\, \text{$\frac{1}{2N}\le x\le 1-\frac{1}{2N}$ and $\frac{1}{2N}\le p\le 1-\frac{1}{2N}$}\rb. \end{equation*}
Thus if $\eps<\frac{1}{2N}$,
\begin{equation*} \left|\hbar_p(\xi^{2,\eps}(\ppp))-\hbar_p(\xi_N(\ppp))\right| \le \vkap\left|\xi^{2,\eps}(\ppp)-\xi_N(\ppp)\right| \end{equation*}
for all $\ppp=(p,\wprho)\in \XX$.  Thanks to \eqref{E:Goofy}, we thus have that $\lim_{\eps \to 0}H(\nu_{N,\eps})=H(\nu)$.

We finally note that $\ppp=(p,\wprho)\mapsto \hbar'_p(\xi^{2,\eps}(\ppp))$ is continuous on $\{\ppp=(p,\wprho)\in \XX: p\in (N^{-1}-\eps,1-N^{-1}+\eps)\}$
(\eqref{E:stuff} ensures that $\tilde \xi_2$ takes values in $(0,1)$ in this case).
On  $\{\ppp=(p,\wprho)\in \XX: p\in (0,1)\setminus (N^{-1}-\eps,1-N^{-1}+\eps)\}$, we have that $\hbar'_p(\xi^{2,\eps}(\ppp))=\hbar'_p(p)=0$.  This finishes the proof.
\end{proof}

\begin{proof}[Proof of Lemma \ref{L:Hrep}]
Assume first that $\nu$ is not absolutely continuous with respect to $\UU$.  We will show that then the right-hand side of \eqref{E:Hrep} is infinite. Then there is an $A\in \Borel(\XX)$ such that $\nu(A)>0$
and $\UU(A)=0$.   Since $\XX$ is Polish, $\nu$ is regular; i.e.,
\begin{equation*} \nu(A) = \sup\lb \nu(F): \text{$F\subset A$, $F$ closed}\rb. \end{equation*}
Thus there is a closed subset $F$ of $A$ such that $\nu(F)>0$.  Fix also now $c>0$.  For each $n\in \N$, define
\begin{equation*} \phi_n(\ppp) \Def c\exp\left[-n \dist(\ppp,F)\right] \qquad \ppp\in \XX \end{equation*}
where $\dist(\ppp,F)$ is the distance (in $\XX$) from $x$ to $F$.  Then $0\le \phi_n\le c$ for all $n\in \N$, and $\phi_n\searrow c\chi_F$.  Recall the definition of $\lambda_p(\theta)$ from (\ref{Eq:Duals_lambda_h}). Since $\theta\mapsto \lambda_p(\theta)$ is
nondecreasing and continuous for each $p\in [0,1]$, we also have that $\lambda_p(\phi_n(\ppp))\searrow \lambda_p(c\chi_F(\ppp))$ for all $\ppp=(p,\wprho)\in \XX$.  Thus
\begin{multline*} \sup_{\phi\in C(\XX)}\lb \int_{\ppp\in \XX}\phi(\ppp)\nu(dp) - \int_{\ppp=(p,\wprho)\in \XX}\lambda_p(\phi(\ppp)) \UU(d\ppp)\rb\\
\ge \varlimsup_{n\to \infty}\lb \int_{\ppp\in \XX}\phi_n(\ppp)\nu(dp) - \int_{\ppp=(p,\wprho)\in \XX}\lambda_p(\phi_n(\ppp)) \UU(d\ppp)\rb
= c\nu(F). \end{multline*}
Let $c\nearrow \infty$ to see that the right-hand side of \eqref{E:Hrep} is infinite.

Assume next that $\nu \ll \UU$.  We use the fact that $\hbar_p$ and $\lambda_p$ are convex duals of each other.  For any $\phi\in C(\XX)$,
\begin{multline*} \int_{\ppp\in \XX}\phi(\ppp)\nu(d\ppp) - \int_{\ppp=(p,\wprho)\in \XX}\lambda_p(\phi(\ppp))\UU(d\ppp)\\
= \int_{\ppp\in \XX}\inf_{x\in \R}\lb \phi(\ppp)\left(\frac{d\nu}{d\UU}(\ppp) - x\right) + \hbar_p(x)\rb \UU(d\ppp)
\le \int_{\ppp=(p,\wprho)\in \XX}\hbar_p\left(\frac{d\nu}{d\UU}(\ppp)\right) \UU(d\ppp). \end{multline*}

To show the reverse inequality, let's write that
\begin{equation*} H(\nu) = \int_{\ppp=(p,\wprho)\in \XX}\sup_{\theta\in \R}\lb \theta \frac{d\nu}{d\UU}(\ppp)-\lambda_p(\theta)\rb \UU(d\ppp)
=\int_{\ppp=(p,\wprho)\in \XX}\lim_{N\to \infty}F_N(\ppp)\UU(d\ppp) \end{equation*}
where
\begin{equation}\label{E:FDef} F_N(\ppp) = \sup_{|\theta|\le N}\lb \theta \frac{d\nu}{d\UU}(\ppp)-\lambda_p(\theta)\rb \end{equation}
for all $N\in \N$ and $\ppp=(p,\wprho)\in \XX$.  We can explicitly solve this minimization problem; for $N\in \N$ and $\ppp=(p,\wprho)\in \XX$, define
\begin{equation*} \phi_N(\ppp) = \begin{cases}  \hbar'_p\left(\frac{d\nu}{d\UU}(\ppp)\right) &\text{if $p\in (0,1)$, $\frac{d\nu}{d\UU}(\ppp)\in (0,1)$, and $-N\le \hbar'_p\left(\frac{d\nu}{d\UU}(\ppp)\right)\le N$} \\
0 &\text{if $p\in \{0,1\}$ and $\frac{d\nu}{d\UU}(\ppp)=p$}\\
N\operatorname{sgn}\left(\frac{d\nu}{d\UU}(\ppp)-p\right) &\text{else}\end{cases}
\end{equation*}
(where $\operatorname{sgn}$ is the standard signum function).
Then
\begin{equation*} F_N(\ppp) = \phi_N(\ppp)\frac{d\nu}{d\UU}(\ppp)-\lambda_p(\phi_N(\ppp)) \end{equation*}
for all $\ppp=(p,\wprho)\in \XX$ and $N\in \N$.  Clearly $F_N$ and $\phi_N$ are measurable, and $\phi_{N}\in B(\XX)$.   From \eqref{E:FDef}, we also see that $F_N$ is nondecreasing in $N$.
Thus by (\ref{E:FDef}) and monotone convergence
\begin{multline*} H(\nu) = \lim_{N\to \infty}\int_{\ppp\in \XX}F_N(\ppp)\UU(d\ppp) = \lim_{N\to \infty}\int_{\ppp=(p,\wprho)\in \XX}\lb \phi_N(\ppp)\frac{d\nu}{d\UU}(\ppp)-\lambda_p(\phi_N(\ppp))\rb \UU(d\ppp)\\
\le \sup_{\varphi\in B(\XX)}\lb \int_{\ppp=(p,\wprho)\in \XX}\phi(\ppp)\nu(d\ppp) -  \int_{\ppp=(p,\wprho)\in \XX}\lambda_p(\phi(\ppp))\UU(d\ppp)\rb. \end{multline*}
Since $\XX$ is Polish, $\nu$ and $\UU$ are regular; and thus we can approximate elements of $B(\XX)$ by elements of $C(\XX)$, completing the proof.
\end{proof}

\bibliographystyle{alpha}
\def\cprime{$'$} \def\cprime{$'$} \def\cprime{$'$} \def\cprime{$'$}
  \def\polhk#1{\setbox0=\hbox{#1}{\ooalign{\hidewidth
  \lower1.5ex\hbox{`}\hidewidth\crcr\unhbox0}}} \def\cprime{$'$}
  \def\cprime{$'$} \def\cprime{$'$} \def\cprime{$'$} \def\cprime{$'$}

\end{document}